%% file: newlbfl-arxiv.tex
\documentclass[11pt]{article} 
\usepackage{times}
\usepackage{amsmath,amsfonts}
\usepackage{epsfig,amssymb,amstext,xspace,theorem}
\usepackage{graphicx,color} 

\newcommand{\np}{{\em NP}\xspace}
\newcommand{\nphard}{\np-hard\xspace} 

\DeclareMathOperator*{\Exp}{E}
\newcommand{\E}[2][{}]{\ensuremath{{\textstyle \Exp_{#1}}\left[#2\right]}}

\newtheorem{theorem}{Theorem}[section]
\newtheorem{lemma}[theorem]{Lemma}

\newtheorem{claim}[theorem]{Claim}

{\theorembodyfont{\rmfamily} \newtheorem{remark}[theorem]{Remark}}

\def\blksquare{\rule{2mm}{2mm}}
\def\qedsymbol{\blksquare}
\newcommand{\bg}[1]{\medskip\noindent{\bf #1}}
\newcommand{\ed}{{\hfill\qedsymbol}\medskip}
\newenvironment{proof}{\bg{Proof : }}{\ed}
\newenvironment{proofof}[1]{\bg{Proof of #1 : }}{\ed}

\newenvironment{proofofnobox}[1]{\bg{Proof of #1 : }}{\medskip}

\newenvironment{labeledthm}[1]{\bg{#1 }\it}{\medskip}

\newcommand{\R}{\ensuremath{\mathbb R}}

\newcommand{\I}{\ensuremath{\mathcal I}}
\newcommand{\F}{\ensuremath{\mathcal F}}
\newcommand{\D}{\ensuremath{\mathcal D}}

\newcommand{\Pc}{\ensuremath{\mathcal P}}
\newcommand{\Qc}{\ensuremath{\mathcal Q}}
\newcommand{\Rc}{\ensuremath{\mathcal R}}
\newcommand{\OPT}{\ensuremath{\mathit{OPT}}}
\newcommand{\cost}{\ensuremath{\mathit{cost}}}

\newcommand{\frall}{\ensuremath{\text{ for all }}}

\newcommand{\sm}{\ensuremath{\setminus}}
\newcommand{\es}{\ensuremath{\emptyset}}

\newcommand{\ceil}[1]{\ensuremath{\left\lceil#1\right\rceil}}
\newcommand{\floor}[1]{\ensuremath{\left\lfloor#1\right\rfloor}}

\newcommand{\e}{\ensuremath{\epsilon}}

\newcommand{\gm}{\ensuremath{\gamma}}

\newcommand{\sse}{\subseteq}

\newcounter{one} \newcounter{two}
\newcommand{\unifl}{{\small \textsf{UniFL}}\xspace}
\newcommand{\ufl}{{\small \textsf{UFL}}\xspace}
\newcommand{\cfl}{{\small \textsf{CFL}}\xspace}

\newcommand{\lbfl}{{\small \textsf{LBFL}}\xspace}
\newcommand{\cdufl}{{\small \textsf{CDUFL}}\xspace}

\newcommand{\hF}{\ensuremath{\widehat F}}
\newcommand{\hC}{\ensuremath{\widehat C}}
\newcommand{\hS}{\ensuremath{\widehat S}}

\newcommand{\htf}{\ensuremath{\widehat f}}
\newcommand{\hc}{\ensuremath{\widehat c}}

\newcommand{\al}{\ensuremath{\alpha}}

\newcommand{\dt}{\ensuremath{\delta}}

\newcommand{\sg}{\ensuremath{\sigma}}

\newcommand{\Gm}{\ensuremath{\Gamma}}

\newcommand{\sol}{\ensuremath{\mathrm{sol}}}
\newcommand{\opt}{\ensuremath{\mathrm{opt}}}

\newcommand{\hcF}{\ensuremath{\widehat{\mathcal F}}}
\newcommand{\hcFc}{\ensuremath{{\textstyle \widehat{\mathcal F}^c}}}
\newcommand{\hcFu}{\ensuremath{{\textstyle \widehat{\mathcal F}^u}}}
\newcommand{\hcD}{\ensuremath{\widehat{\mathcal D}}}
\newcommand{\hcDS}{\ensuremath{{\textstyle \widehat{\mathcal D}_{\hS}}}}
\newcommand{\hcDsol}{\ensuremath{{\textstyle \widehat{\mathcal D}_{\sol}}}}
\newcommand{\Pcst}[1]{\ensuremath{\Pc^{\mathsf{st}}({#1})}}
\newcommand{\Pcend}[1]{\ensuremath{\Pc^{\mathsf{end}}({#1})}}
\newcommand{\head}{\ensuremath{\mathit{head}}}
\newcommand{\tail}{\ensuremath{\mathit{tail}}}
\newcommand{\add}{\ensuremath{\mathrm{add}}}
\newcommand{\drop}{\ensuremath{\mathrm{delete}}}
\newcommand{\swap}{\ensuremath{\mathrm{swap}}}
\newcommand{\capt}{\ensuremath{\mathrm{capt}}}
\newcommand{\shift}{\ensuremath{\mathit{shift}}}

\newcommand{\hh}{\ensuremath{h}}
\newcommand{\ch}{\ensuremath{c}}
\newcommand{\Kh}{\ensuremath{K}}

\newenvironment{labellist}[2][-2ex]
{\begin{list}{{#2}\arabic{enumi}.}{\usecounter{enumi} \addtolength{\leftmargin}{#1}}}
{\end{list}}

\title{Improved Approximation Guarantees for Lower-Bounded Facility Location \\
{\large (Extended Abstract)}}
\author{
         Sara Ahmadian\thanks{{\tt \{sahmadian,cswamy\}@math.uwaterloo.ca}. 
         Dept. of Combinatorics and Optimization, Univ. Waterloo, Waterloo, ON N2L 3G1.
         Supported in part by NSERC grant 327620-09. The second author is also supported
         by an Ontario Early Researcher Award.}  
\and
\addtocounter{footnote}{-1}
         Chaitanya Swamy\footnotemark
}

\date{}

\begin{document}

\maketitle
\def\thepage{}
\thispagestyle{empty}

\begin{abstract}
We consider the {\em lower-bounded facility location} (\lbfl) problem (also sometimes
called {\em load-balanced facility location}), which is a generalization of 
{\em uncapacitated facility location} (\ufl), where each open facility is required to
serve a certain {\em minimum} amount of demand. More formally, an instance $\I$ of \lbfl
is specified by a set $\F$ of facilities with facility-opening costs $\{f_i\}$, a set $\D$
of clients, and connection costs $\{c_{ij}\}$ specifying the cost of assigning a client
$j$ to a facility $i$, where the $c_{ij}$s form a metric. A feasible solution specifies a 
subset $F$ of facilities to open, and assigns each client $j$ to an open facility $i(j)\in F$ so
that each open facility serves {\em at least $M$ clients}, where $M$ is an input
parameter. The cost of such a solution is $\sum_{i\in F}f_i+\sum_j c_{i(j)j}$, and the
goal is to find a feasible solution of minimum cost.   

The current best approximation ratio for \lbfl is $448$~\cite{Svitkina08}.  
We substantially advance the state-of-the-art for \lbfl by devising an approximation
algorithm for \lbfl that achieves a significantly-improved approximation guarantee of
$82.6$. 

Our improvement comes from a variety of ideas in algorithm design and
analysis, which also yield new insights into \lbfl. 
Our chief algorithmic novelty is to present an improved method for solving a
more-structured \lbfl instance obtained from $\I$ via a bicriteria approximation algorithm
for \lbfl, wherein all clients are aggregated at a subset $\F'$ of facilities, each having at
least $\al M$ co-located clients (for some $\al\in[0,1]$).  
One of our key insights is that one can reduce the resulting \lbfl instance, denoted
$\I_2(\al)$, to a problem we introduce, called {\em capacity-discounted \ufl}
(\cdufl). \cdufl is a special case of capacitated facility location (\cfl) where
facilities are either uncapacitated, or have finite capacity and zero opening
costs. Circumventing the difficulty that \cdufl inherits the intractability of \cfl with
respect to LP-based approximation guarantees, we give a simple local-search algorithm for
\cdufl based on add, delete, and swap moves that achieves the same approximation ratio (of
$1+\sqrt{2}$) as the corresponding local-search algorithm for \ufl. In contrast, the
algorithm in~\cite{Svitkina08} proceeds by reducing $\I_2(\al)$ to \cfl, whose
current-best approximation ratio is worse than that of our local-search algorithm for
\cdufl, and this is one of the reasons behind our algorithm's improved approximation 
ratio.  

Another new ingredient of our \lbfl-algorithm and analysis
is a subtly different method for constructing a bicriteria solution for $\I$ (and hence, 
$\I_2(\al)$), combined with the more significant change that we now choose a {\em random} 
$\al$ from a suitable distribution. 
This leads to a surprising degree of improvement in the approximation factor, which 
is reminiscent of the mileage provided by random $\al$-points in scheduling problems. 
\end{abstract}  

\newpage
\pagenumbering{arabic}
\normalsize

\section{Introduction}
Facility location problems have been widely studied in the Operations Research 
community (see, e.g.,~\cite{MirchandaniF90}). 
In its simplest version, {\em uncapacitated facility location} (\ufl), we are given a
set of facilities 
with opening costs, and a set of clients, 
and we want to open some facilities 
and assign each client to an open facility so as to minimize the sum of the
facility-opening and client-assignment costs.   
This problem 
has a wide range of applications. For example, a company might want to open its 
warehouses at some locations so that its total cost of opening warehouses and 
servicing customers is minimized. 

We consider the {\em lower-bounded facility location} (\lbfl) problem, which is a
generalization of \ufl 
where each open facility
is required to serve a certain {\em minimum} amount of demand. More formally, an \lbfl
instance $\I$ is specified by a set $\F$ of facilities, and a set $\D$ of clients. Opening
facility $i$ incurs a {\em facility-opening cost} $f_i$, and assigning a client $j$ to a
facility $i$ incurs a {\em connection cost} $c_{ij}$. A feasible solution specifies a
subset $F\sse\F$ of facilities, and assigns each client $j$ to an open facility 
$i(j)\in F$ so that {\em each open facility serves at least $M$ clients}, where $M$ is an
input parameter. The cost of such a solution is the sum of the facility-opening and
connection costs, that is, $\sum_{i\in F}f_i+\sum_j c_{i(j)j}$, and the goal is to find a
feasible solution of minimum cost. 
As is standard in the study of facility location problems, we assume throughout that
$c_{ij}$s form a metric. We use the terms connection cost and assignment cost
interchangeably in the sequel.

\lbfl can be motivated from various perspectives.
This problem was introduced independently by Karger and
Minkoff~\cite{KargerM00}, and Guha, Meyerson, and Munagala (who called the problem  
{\em load-balanced facility location})~\cite{GuhaMM00} (see also~\cite{GuhaMM01b}), 
both of whom arrived at \lbfl as a means of solving their respective buy-at-bulk style
network design problems. 
\lbfl arises as a natural subroutine in such settings
because obtaining a near-optimal solution to the buy-at-bulk problem often entails
aggregating a certain minimum demand at certain hub locations, and then connecting
the hubs via links of lower per-unit-demand cost (and higher fixed cost).
\lbfl also finds direct applications in 
supply-chain logistics problems, where the lower-bound constraint can be used to model the
fact that it is not profitable or feasible to use services 
unless they satisfy a certain minimum demand. 
For example (as noted in~\cite{Svitkina08}), Lim, Wang, and Xu~\cite{LimWX06}, use \lbfl to
abstract a transportation problem faced by a company that has to determine the allocation
of cargo from customers to carriers, who then ship their cargo overseas. Here the lower
bound arises because each carrier, if used, is required (by regulation) to deliver a
minimum amount of cargo. 
Also, \lbfl is an interesting special case of {\em universal facility location}
(\unifl)~\cite{MahdianP03}---a generalization of \ufl where the facility cost 
depends on the number of clients served by it---with non-increasing facility-cost functions.    
\unifl with arbitrary non-increasing 
functions is not a well-understood problem, 
\nolinebreak \mbox
{and the study of \lbfl may provide useful insights here.}

Clearly, \lbfl with $M=1$ is simply \ufl, and hence, is \nphard; consequently, we are 
interested in designing approximation algorithms for \lbfl. 
The first constant-factor approximation algorithm for \lbfl was devised
by Svitkina~\cite{Svitkina08}, whose approximation ratio is $448$. 
Prior to this, the only known approximation guarantees were {\em bicriteria guarantees}.  
\cite{KargerM00} and~\cite{GuhaMM00} independently devised 
$(\rho,\al)$-approximation algorithms via a reduction to \ufl:  
these algorithms return a solution of cost at most $\rho$ times the optimum where each 
open facility serves at least $\al M$ clients
\nolinebreak \mbox
{($\al<1$, $\rho$ is a function of $\al$).} 

\vspace{-1ex}
\paragraph{Our results and techniques.}
We devise an approximation algorithm for \lbfl that achieves a substantially-improved
approximation guarantee of $82.6$ (Theorem~\ref{mainthm}), thus significantly advancing
the state-of-the-art for \lbfl. Our improvement comes from a combination of ideas in
algorithm design and analysis, and yields new insights about the approximability of
\lbfl. 
In order to describe the ideas underlying our improvement, we first briefly sketch
Svitkina's algorithm.  

Svitkina's algorithm begins by using the reduction in~\cite{KargerM00,GuhaMM00} to obtain
a bicriteria solution for $\I$, which is then used to convert $\I$
into an \lbfl instance $\I_2$ with facility-set $\F'\sse\F$ having the following
structure: 
(i) all clients are aggregated at $\F'$ 
with each facility $i\in\F'$ having $n_i\geq \al M$ co-located clients;  
(ii) all facilities in $\F'$ have zero opening costs; and 
(iii) near-optimal solutions to $\I_2$ translate to near-optimal solutions to $\I$ 
(and vice versa).
The goal now is to identify a subset of $\F'$ 
to close, such 
that transferring the clients aggregated at these closed facilities to the remaining
(open) facilities in $\F'$ ensures that each remaining facility serves at least $M$ demand
(and the cost incurred is ``small''). \cite{Svitkina08} shows that one can achieve this by
solving a suitable \cfl instance. Essentially the idea is that a facility $i$ that remains
open corresponds to a {\em demand point} in the \cfl instance that requires $M-n_i$ units
of demand, and a facility $i$ that is closed maps to a {\em supply point} in the \cfl
instance having $n_i$ units that can be supplied to demand points (i.e., open facilities).  
Of course, one does not know beforehand which facilities will be closed and which
will remain open; 
so to encode this correspondence in the \cfl instance, we create at every location 
$i\in\F'$, a supply point with (suitable opening cost and) capacity $M$, and
a demand point with demand $M-n_i$ if $n_i\leq M$ 
(so the supply point at $i$ has $n_i$ residual capacity after satisfying this demand). 
(Assume $n_i\leq M$ for simplicity; facilities with $n_i>M$ are treated differently.)   
Finally, \cite{Svitkina08} argues that a \cfl-solution (where a supply point may end up
sending {\em less} then $n_i$ supply to other demand points) can be mapped to a solution
to $\I_2$ without increasing the cost incurred by much; since \cfl admits an
$O(1)$-approximation algorithm, one obtains an $O(1)$-approximate solution to 
$\I_2$, and hence to the original \lbfl instance $\I$. 

Our algorithm also proceeds by (a) obtaining an 
\lbfl instance $\I_2$ satisfying properties (i)--(iii) mentioned above, (b) solving $\I_2$,
and (c) mapping the $\I_2$-solution to a solution to $\I$, but our implementation of steps 
(a) and (b) differs from that in Svitkina's algorithm.
These modified implementations, which are independent of each other and yield significant
improvements in the overall approximation ratio even when considered in isolation, result
in our much-improved approximation ratio.    
We detail how we perform step (a) later, and focus first on describing how we solve 
$\I_2$, which is our chief algorithmic contribution. 

Our key insight is that 
one can solve the \lbfl instance $\I_2$ by reducing it to a new problem we
introduce that we call {\em capacity-discounted} \ufl (\cdufl), which closely resembles
\ufl and admits an algorithm (that we devise) with a much better approximation ratio than
\cfl. 
A \cdufl-instance has the property that every
facility is either uncapacitated (i.e., has infinite capacity), or has finite capacity and
{\em zero} facility cost. 
The \cdufl instance we construct consists of the same supply and demand points as in the 
reduction of $\I_2$ to \cfl in~\cite{Svitkina08}, except that all supply  
points with non-zero opening cost are now uncapacitated. 
(An interesting consequence is that if all facilities in $\I_2$ have
$n_i\leq M$, the \cdufl instance is in fact a \ufl-instance!) 

We prove two crucial algorithmic results.
It is not hard to see that the ``standard'' integrality-gap example for the natural
LP-relaxation of \cfl can be cast as a \cdufl instance, thus showing that the natural
LP-relaxation for \cdufl has a large integrality gap (see Appendix~\ref{cduflgap}); in 
fact, we are not aware of any LP-relaxation for \cdufl with constant integrality gap. 
Circumventing this difficulty, we devise a local-search algorithm for \cdufl based on add, 
swap, and delete moves that achieves the {\em same performance guarantees} as the
corresponding local-search algorithm for \ufl~\cite{AryaGKMMP01} (see
Section~\ref{cdufl-ls}).   
The local-search algorithm yields significant dividends in the overall approximation ratio
because not only is its approximation ratio for \cdufl better than the state-of-the-art
for \cfl, but also because it yields separate (asymmetric) guarantees on the
facility-opening and assignment costs, which allows one to perform a tighter analysis. 
Second, we show that any near-optimal \cdufl-solution can be mapped to a near-optimal
solution to $\I_2$ (see Section~\ref{i2-soln}). 
As before, it could be that in the \cdufl-solution, a supply point $i$ (which corresponds
to facility $i$ being closed down) sends less than $n_i$ supply to other demand points, so
that closing down $i$ entails transferring its residual clients to open facilities. But
since some supply points are now uncapacitated, it could also be that $i$ sends more than
$n_i$ supply to other demand points. We argue that this artifact can also be handled
without increasing the solution cost by much, by opening the facilities in a
carefully-chosen subset of $\{i\}\cup\{\text{demand points satisfied by $i$}\}$ and
closing down the remaining facilities. 
For {\em every value of $\al$} (recall that the \lbfl instance $\I_2$ is specified in
terms of a parameter $\al$), the resulting approximation factor for $\I_2$
(Theorem~\ref{i2thm})   
is better than the guarantee obtained for $\I_2$ in Svitkina's algorithm; this in turn
translates (by choosing $\al$ suitably) to an improved solution to the original instance. 

We now discuss how we implement step (a), that is, how we obtain instance $\I_2$.
As in~\cite{Svitkina08}, we arrive at $\I_2$ by computing a bicriteria solution to \lbfl,  
but we obtain this bicriteria solution in a different fashion 
(see Section~\ref{bicriteria}).  
The reduction from \lbfl to \ufl in~\cite{KargerM00,GuhaMM00} proceeds by setting the opening
cost of facility $i$ to $f_i+\frac{2\al}{1-\al}\cdot\sum_{j\in\D(i)} c_{ij}$, where
$\D(i)$ is the set of $M$ clients closest to $i$, solving the resulting \ufl instance, and
postprocessing using (single-facility) delete moves if such a move improves the solution
cost. 
We modify this reduction subtly by creating a \ufl instance, where facility $i$'s
opening cost is instead set to $f_i+2\al M R_i(\al)$, where $R_i(\al)$ is the distance
between $i$ and the $\al M$-closest client to it.  
As in the case of the earlier reduction, we argue that each open facility $i$ in the
resulting solution (obtained by solving \ufl and postprocessing) serves at least $\al M$  
clients. 
The overall bound we obtain on the total cost now includes various $R_i(\al)$ terms. 
Instead of plugging in the (weak) bound $M R_i(\al)\leq\frac{\sum_{j\in\D(i)}c_{ij}}{1-\al}$  
(which would yield the same guarantee as that obtained via the earlier reduction), 
we are able to perform a tighter analysis by choosing $\al$ from a suitable
distribution and leveraging the fact that 
$M\int_0^1 R_i(\al)d\al=\sum_{j\in\D(i)}c_{ij}$. 
(This can easily be derandomized, since there are only $M$ combinatorially distinct
choices for $\al$.) 
These simple modifications (in algorithm-design {\em and} analysis) yield a surprising 
amount of improvement in the approximation factor, which is reminiscent of the mileage
provided by (random) $\al$-points for various scheduling problems (see,
e.g.,~\cite{Skutella06}) and \ufl~\cite{ShmoysTA97,Sviridenko02}.    
Also, we observe that one can obtain further improvements by using the local-search 
algorithm of~\cite{CharikarG99,AryaGKMMP01} to solve the above \ufl instance: 
this is because the resulting solution is then already postprocessed, which allows us to 
exploit the asymmetric bounds on the facility-opening and assignment costs provided by the 
local-search algorithm via scaling, and improve the approximation ratio.  

Finally, we remark that the study of \cdufl may provide useful and interesting insights
about \cfl. 
\cdufl is a special case of \cfl that despite its special structure 
inherits the intractability of \cfl with respect to LP-based approximation guarantees. If
one seeks to develop LP-based techniques and algorithms for \cfl (which has been a
long-standing and intriguing open question), then one needs to understand how one
can leverage LP-based techniques for \cdufl, and it is plausible that LP-based insights 
developed for \cdufl 
may yield similar insights for \cfl (and potentially LP-based approximation
guarantees for \cfl).   

\vspace{-1ex}
\paragraph{Related work.}
As mentioned earlier, \lbfl was independently introduced by~\cite{KargerM00}
and~\cite{GuhaMM00}, who 
used it as a subroutine to solve the ({\em rent-or-buy} and hence, the) {\em maybecast}
problem, and the {\em access network design} problem respectively. 
Their ideas, which lead to bicriteria guarantees for \lbfl, play a preprocessing role both  
in Svitkina's algorithm for \lbfl~\cite{Svitkina08} and (slightly indirectly) in our
algorithm. 

There is a large body of literature that deals with approximation algorithms for (metric)
\ufl, \cfl and its variants; see~\cite{Shmoys04} for a survey on \ufl. The first constant
approximation guarantee for \ufl was obtained by Shmoys, Tardos, and
Aardal~\cite{ShmoysTA97} via an LP-rounding algorithm, and the current state-of-the-art is
a 1.488-approximation algorithm due to Li~\cite{Li11}. 
Local-search techniques have also been utilized to obtain $O(1)$-approximation guarantees
for \ufl~\cite{KorupoluPR00,CharikarG99,AryaGKMMP01}. We apply some of the ideas
of~\cite{CharikarG99,AryaGKMMP01} in our algorithm.  
Starting with the work of Korupolu, Plaxton, and Rajaraman~\cite{KorupoluPR00}, various
local-search algorithms with constant approximation ratios have been devised for \cfl,
with the current-best approximation ratio being $5.83+\e$~\cite{ZhangCY03}. 
Local-search approaches are however not known to work for \lbfl; in
Appendix~\ref{app:LBFLlocgap}, we show that local search based on $\add$, $\drop$, and
$\swap$ moves yields poor approximation guarantees.
Universal facility location (\unifl), where the facility cost is a non-decreasing function
of the number of clients served by it, was 
introduced by~\cite{HajiaghayiMM03,MahdianP03},  
and~\cite{MahdianP03} gave a constant approximation algorithm for this. We are not aware
of any work on \unifl with arbitrary non-increasing functions (which generalizes
\lbfl). \cite{GuhaMM00b} give a 
constant approximation for the case where the cost-functions do not decrease too steeply
(the constant depends on the steepness); notice that \lbfl does not fall
\nolinebreak \mbox
{into this class so their results do not apply here.}

\section{Problem definition and notation}
Recall that we have a set $\F$ of facilities with facility-opening costs $\{f_i\}$,
a set $\D$ of clients, metric connection (or assignment) costs $\{c_{ij}\}$ specifying the
cost of assigning client $j$ to facility $i$, and a (integer) parameter $M$. Our objective
is to open a subset $F$ of facilities and assign each client $j$ to an open facility
$i(j)\in F$, so that  at least $M$ clients are assigned to each open facility, and the
total cost incurred, $\sum_{i\in F}f_i+\sum_j c_{i(j)j}$, is minimized. 
We use $\I$ to denote this \lbfl instance. 

Let $F^*$ and $C^*$ denote respectively the facility-opening and assignment cost of an
optimal solution to $\I$; we will often refer to this solution as ``the optimal solution''
in the sequel. We sometimes abuse notation and also use $F^*$ to denote
the set of open facilities in this optimal solution. 
Let $\OPT=F^*+C^*$ denote the total optimal cost. 
For a facility $i\in\F$, let $\D(i)$ be the set of $M$ clients closest to $i$,  
and $R_i(\al)$ denote the distance between $i$ and the $\ceil{\al M}$-closest client to
$i$; that is, if $\D(i)=\{j_1,\ldots,j_M\}$, where 
$c_{ij_1}\leq\ldots\leq c_{ij_M}$, then $R_i(\al)=c_{ij_{\ceil{\al M}}}$ 
(for $0<\al\leq 1$). Let $R^*(\al)=\sum_{i\in F^*}R_i(\al)$. 
Observe that each $R_i(\al)$ is an increasing function of $\al$, 
$M\int_0^1 R_i(\al)d\al=\sum_{j\in\D(i)}c_{ij}$, and
$R_i(\al)\leq\bigl(\sum_{j\in\D(i)}c_{ij})/(M-\ceil{\al M}+1)
\leq\frac{\sum_{j\in\D(i)}c_{ij}}{M(1-\al)}$. 
Hence, $R^*(\al)$ is an increasing function of $\al$,   
$M\int_0^1 R^*(\al)d\al\leq C^*$, and $R^*(\al)\leq\frac{C^*}{M(1-\al)}$.

\section{Our algorithm and the main theorem}  \label{bicriteria}
We now give a high-level description of our algorithm using certain building blocks
that are supplied in the subsequent sections.
Let $\I$ denote the \lbfl instance.

\begin{list}{(\arabic{enumi})}{\usecounter{enumi} \addtolength{\leftmargin}{-2ex} \topsep=0.5ex}
\item {\bf Obtaining a bicriteria solution.\ }
Construct a \ufl instance with the same set of facilities and clients, and the same
assignment costs as $\I$, where the opening cost of facility $i$ is set to 
$f_i+2\al MR_i(\al)$. 
Use the local-search algorithm for \ufl in~\cite{CharikarG99} or~\cite{AryaGKMMP01}
with scaling parameter $\gm>0$ to solve this \ufl instance. 
(We set $\al, \gm$ suitably to get the desired approximation; see Theorem~\ref{mainthm}.)
Let $\F'\sse\F$ be the set of facilities opened in the \ufl-solution.
Claim~\ref{bicrit-lb} and Lemma~\ref{bicrit-perf} 
\nolinebreak\mbox
{show that each $i\in F'$ serves at least $\al M$ clients.}

\item {\bf Transforming to a structured \lbfl instance.\ }
We use the bicriteria solution obtained above to transform $\I$ into another structured
\lbfl instance $\I_2$ as in~\cite{Svitkina08}.  
In the instance $\I_2$, we set the opening cost of each
$i\in\F'$ to zero, and we ``move'' to $i$ all the $n_i\geq\al M$ clients assigned to it,
that is, all these clients are now co-located at $i$. So $\I_2$ consists of only the
points in $\F'$ (which forms both the facility-set and client-set). We will sometimes use
the notation $\I_2(\al)$ to indicate explicitly that $\I_2$'s specification depends on the
parameter $\al$.

\item Solve $\I_2$ using the method described in Section~\ref{i2solve}. Obtain a solution
to $\I$ by opening the same facilities and making the same client assignments as in the
solution to $\I_2$. 
\end{list}

\vspace{-1ex}
\paragraph{Analysis.}
Our main theorem is as follows. 

\begin{theorem} \label{mainthm}
For any $\al\in(0.5,1]$ and $\gm>0$, the above algorithm returns a solution to $\I$ of
cost at most \vspace{-4pt} 
$$
F^*\bigl(1+\gm\hh(\al)\bigr)+
C^*\Bigl(2\hh(\al)-1+\tfrac{2}{\gm}\Bigr)+
2\gm\al M R^*(\al)\hh(\al)+2\al MR^*(\al) \\[-6pt] $$
where $\hh(\al)=1+\frac{4}{\al}+\frac{4\al}{2\al-1}+4\sqrt{\frac{6}{2\al-1}}$.
Thus, we can compute efficiently a solution to $\I$ of cost at most:
\begin{list}{(\roman{enumii})}{\usecounter{enumii} \topsep=0ex \itemsep=0ex
    \addtolength{\leftmargin}{-4ex}}
\item $92.84\cdot\OPT$, by setting $\al=0.75, \gm=3/\hh(\al)$;
\item $82.6\cdot\OPT$, by letting $\gm$ be a suitable (efficiently-computable) function of 
$\al$, and choosing $\al$ randomly from the interval $[0.67,1]$ 
according to the density function $p(x)=\frac{1}{\ln(1/0.67)x}$. 
\end{list}
\end{theorem} 

The roadmap for proving Theorem~\ref{mainthm} is as follows. We first bound the cost of
the bicriteria solution obtained in step (1) in terms of $\OPT$
(Lemma~\ref{bicrit-perf}). This will allow us to bound the cost of an optimal solution to
$\I_2$, and argue that mapping an $\I_2$-solution to a solution to $\I$ does not increase the
cost by much (Lemma~\ref{i2mapping}). The only missing ingredient is a guarantee on the
cost of the solution to $\I_2$ found in step (3), which we supply in Theorem~\ref{i2thm},
whose proof appears in Section~\ref{i2solve}. 

The following claim follows from essentially the same arguments as 
in~\cite{KargerM00,GuhaMM00}.

\begin{claim} \label{bicrit-lb}
Let $S'$ be a {\em delete-optimal} solution to the above \ufl instance; that is, the
total \ufl-cost does not decrease by deleting any open facility of $S'$. Then, each
facility of $S'$ serves at least $\al M$ clients.
\end{claim} 

The local-search algorithms for \ufl in~\cite{CharikarG99,AryaGKMMP01} have the same
performance guarantees and both include a delete-move as a local-search operation, so upon
termination, we obtain a delete-optimal solution.\footnote{A subtle point is that typically
local-search algorithms terminate only with an ``approximate'' local optimum. However, one 
can then execute all delete moves that improve the solution cost, and thereby obtain a
delete-optimal solution.} 
Observe that opening the same facilities and making the same client assignments as in the
optimal solution to $\I$ yields a solution $S$ to the \ufl instance constructed in step
(1) of the algorithm with facility cost $F^S\leq F^*+2\al M R^*(\al)$ and assignment cost
$C^S\leq C^*$.   
Combined with the analysis in~\cite{CharikarG99,AryaGKMMP01}, this yields the following. 
(For simplicity, we assume that all local-search algorithms return a local optimum;
standard arguments show that 
dropping this assumption increases the approximation by at most a $(1+\e)$ factor.)

\begin{lemma} \label{bicrit-perf}
For a given parameter $\gm>0$, executing the local-search algorithm
in~\cite{CharikarG99,AryaGKMMP01} on the above \ufl instance returns a solution with
facility cost $F_b$ and assignment cost $C_b$ satisfying
$F_b\leq F^*+2\al MR^*(\al)+2C^*/\gm,\ C_b\leq\gm\bigl(F^*+2\al MR^*(\al)\bigr)+C^*$,
where each open facility serves at least $\al M$ clients. 
\end{lemma}

\begin{lemma}[~\cite{Svitkina08}] \label{i2mapping}
(i) The (assignment) cost $C^*_{\I_2}$ of an optimal solution to $\I_2$ is at most
$2(C_b+C^*)$. 

\noindent
(ii) Any solution to $\I_2$ of cost $C$ yields a solution to $\I$ of cost at most
$F_b+C_b+C$. 
\end{lemma}

\begin{theorem} \label{i2thm}
For any $\al>0.5$, there is a $g(\al)$-approximation algorithm for $\I_2(\al)$, where  
$g(\al)=\frac{2}{\al}+\frac{2\al}{2\al-1}+2\sqrt{\frac{2}{\al^2}+\frac{4}{2\al-1}}$.
\end{theorem}

\begin{remark} \label{i2remk}
Our $g(\al)$-approximation ratio for $\I_2(\al)$ improves upon the approximation obtained 
in~\cite{Svitkina08} by a factor of roughly 2 {\em for all $\al$}. Thus, plugging in
our algorithm for solving $\I_2$ in the \lbfl-algorithm in~\cite{Svitkina08} (and choosing
a suitable $\al$), already yields an improved approximation factor of $218$ for \lbfl. 
\end{remark}

\begin{proofofnobox}{Theorem~\ref{mainthm}}
Recall that $h(\al)=1+\frac{4}{\al}+\frac{4\al}{2\al-1}+4\sqrt{\frac{6}{2\al-1}}$.
Note that $2g(\al)+1\leq h(\al)$ for all $\al\in[0,1]$; we use this upper bound throughout
below. 
Combining Theorem~\ref{i2thm} and the bounds in Lemmas~\ref{bicrit-perf}
and~\ref{i2mapping}, we obtain a solution to $\I$ of cost at most 
$F_b+\bigl(2g(\al)+1\bigr)C_b+2g(\al)C^*$
\begin{eqnarray*}
& \leq & F^*+2\al MR^*(\al)+\frac{2C^*}{\gm}+\hh(\al)\gm
\Bigl(F^*+2\al MR^*(\al)\Bigr)+\bigl(2\hh(\al)-1\bigr)C^* \\
& = & F^*\bigl(1+\gm\hh(\al)\bigr)+
C^*\Bigl(2\hh(\al)-1+\tfrac{2}{\gm}\Bigr)+2\gm\al MR^*(\al)\hh(\al)+2\al MR^*(\al).
\end{eqnarray*}

Part (i) follows by plugging in the values of $\al$ and $\gm$, and using the 
bound $R^*(\al)\leq\frac{C^*}{M(1-\al)}$. 

Let $\beta=0.67$. For part (ii), we set $\gm=\frac{\Kh}{\sqrt{\hh(\al)}}$, where  
$\Kh=\left(\ln^2(1/\beta)\cdot\E[\al]{\hh(\al)}/
\bigl(\frac{\int_{\beta}^1\hh(x)dx}{1-\beta}\bigr)\right)^{\frac{1}{4}}$. Plugging in this
$\gm$, we see that the cost incurred is at most   
$$
F^*\bigl(1+\Kh\sqrt{\hh(\al)}\bigr)+
C^*\Bigl(2\hh(\al)-1+\tfrac{2}{\Kh}\sqrt{\hh(\al)}\Bigr)+
2\Kh\al M R^*(\al)\sqrt{\hh(\al)}+2\al MR^*(\al). 
$$
We now bound the expected cost incurred when one chooses $\al$ randomly according to the 
stated density function. 
This will also yield an explicit expression for $\Kh$ (as a function of $\beta$), thus
showing that $K$ (and hence, $\gm$) can be computed efficiently.
We note that $\E{\sqrt{X}}\leq\sqrt{\E{X}}$ and utilize Chebyshev's Integral inequality
(see~\cite{HardyLP52}): if $f$ and $g$ are non-increasing and non-decreasing functions
respectively from $[a,b]$ to $\R_+$, then  
$\int_a^b f(x)g(x)dx\leq\frac{(\int_a^b f(x)dx)(\int_a^b g(x)dx)}{b-a}$. 
Observe that $h(\al)$ decreases with $\al$. Recall that $\beta=0.67$.
We have the following.
\begin{eqnarray*}
\E[\al]{\hh(\al)} & = & \ch_2(\beta):=
\Bigl[\frac{4}{\beta}-4+8\sqrt{6}\bigl(\pi/4-\tan^{-1}(\sqrt{2\beta-1})\bigr)+
2\ln\Bigl(\frac{1}{2\beta-1}\Bigr)+\ln(1/\beta)\Bigr]/\ln(1/\beta) \\[0.2ex]
\E[\al]{\al MR^*(\al)} & = & M\Bigl(\int_{\beta}^1 R^*(x)dx\Bigr)/\ln(1/\beta)
\leq C^*/\ln(1/\beta).
\end{eqnarray*}
Finally, using Chebyshev's inequality, we obtain that
$$
\E[\al]{\al M R^*(\al)\sqrt{\hh(\al)}} \leq 
\Bigl[M\Bigl(\int_{\beta}^1 R^*(x)dx\Bigr)
\tfrac{\int_{\beta}^1 dx\sqrt{\hh(x)}}{1-\beta}\Bigr]/\ln(1/\beta) 
\leq \Bigl[C^*\sqrt{\ch_3(\beta)}\Bigr]/\ln(1/\beta),
$$
where 
$$ 
\ch_3(\beta):=\bigl(\int_{\beta}^1 \hh(x)dx\bigr)/(1-\beta)=
\Bigl[4\ln(1/\beta)+4\sqrt{6}\bigl(1-\sqrt{2\beta-1}\bigr)+3(1-\beta)+
\ln\Bigl(\frac{1}{2\beta-1}\Bigr)\Bigr]/(1-\beta).
$$
The second inequality follows since 
$\bigl(\int_{\beta}^1 dx\sqrt{\hh(x)}\bigr)/(1-\beta)
=\E[\al\sim\text{uniform in $[\beta,1]$}]{\sqrt{\hh(\al)}}$. 
Plugging in these bounds, we get that 
$\Kh=\bigl(\ln^2(1/\beta)\ch_2(\beta)/\ch_3(\beta)\bigr)^{0.25}$ and 
the total cost is at most
$$
F^*\Bigl(1+
\bigl(\tfrac{\ln^2(1/\beta)(\ch_2(\beta))^3}{\ch_3(\beta)}\bigr)^{\frac{1}{4}}\Bigr)+
C^*\Bigl(2\ch_2(\beta)-1+
4\bigl(\tfrac{\ch_2(\beta)\ch_3(\beta)}{\ln^2(1/\beta)}\bigr)^{\frac{1}{4}}
+\tfrac{2}{\ln(1/\beta)}\Bigr)
< 82.59(F^*+C^*). \qquad \qedsymbol \\[-20pt]
$$
\end{proofofnobox}

\section{Solving instance \boldmath $\I_2(\al)$} \label{i2solve}
We now describe our algorithm for solving instance $\I_2(\al)$ and analyze its performance 
guarantee, thereby proving Theorem~\ref{i2thm}. As mentioned earlier, one of the key
differences between our algorithm and the one in~\cite{Svitkina08} is that instead of
reducing $\I_2$ to capacitated facility location (\cfl), we solve $\I_2$ by reducing it to
a new problem that we call {\em capacity-discounted \ufl} (\cdufl). 
\cdufl is a special case of \cfl where all facilities with non-zero opening cost are
uncapacitated (i.e., have infinite capacity). Perhaps surprisingly, despite this 
special structure, \cdufl inherits the intractability of \cfl 
with respect to LP-based approximation guarantees: 
there is no known LP-relaxation for \cdufl that has constant integrality gap; 
Appendix~\ref{cduflgap} shows that the natural LP-relaxation for \cdufl has bad
integrality gap.
However, as we show in Section~\ref{cdufl-ls}, we can obtain a simple local-search
algorithm for \cdufl  
whose approximation ratio is better than the current-best approximation for \cfl. 

Recall that $\I_2$ has only the points in $\F'\sse\F$, and there are 
$n_i\geq\al M$ co-located clients at each $i\in\F'$. 
Let $l(i)=\min_{i'\in\F',i'\neq i}c_{ii'}$. 
To avoid confusion, we refer to the facilities and clients in the
\cdufl instance as supply points and demand points respectively. The \cdufl instance
created to solve $\I_2$ resembles the \cfl instance created in~\cite{Svitkina08}; the
difference is that all supply points with non-zero opening costs are now uncapacitated.  
More precisely, at each $i\in\F'$, we create an uncapacitated supply point with opening
cost $\dt\min\{n_i,M\}l(i)$, where $\dt$ is a parameter we fix later. If $n_i>M$ we create
a second supply point at $i$ with capacity $n_i-M$ and zero opening cost. If $n_i<M$, we
create a demand point at $i$ with demand $M-n_i$. 
Let $\I'$ denote this \cdufl instance (see Fig.~\ref{I2toIp}).
Let $\F^u,\ \F^c$ denote respectively the set of uncapacitated and capacitated supply
points of $\I'$.  
Roughly speaking, satisfying a demand point $i$ by non-co-located supply points
translates to leaving facility $i$ open in the $\I_2$ solution; hence, its demand is set
to $M-n_i$, which is the number of additional clients it needs. Conversely, opening
the uncapacitated supply point at $i$ and supplying demand points from $i$ translates to    
closing $i$ in the $\I_2$ solution and transferring its co-located clients to other open
facilities.  

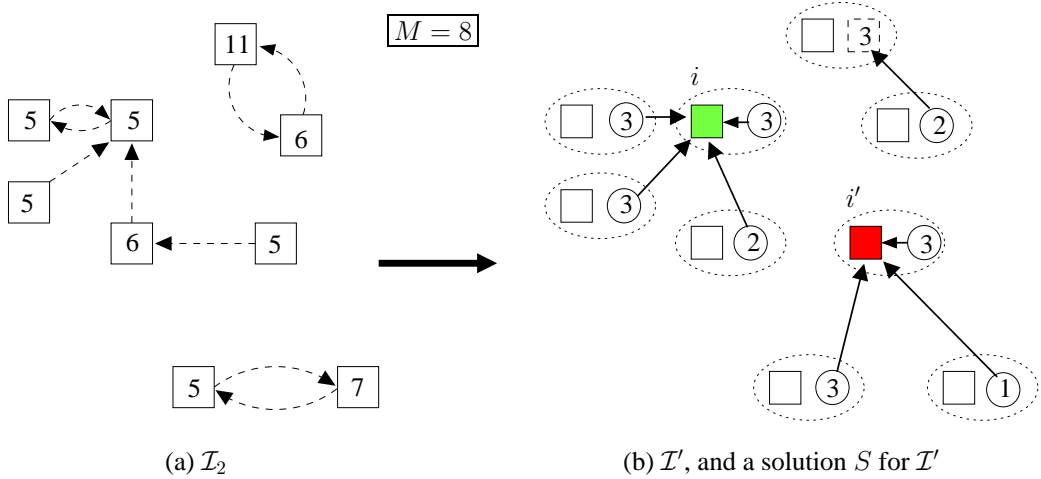
\begin{figure}[ht!]
\centerline{\resizebox{!}{2.5in}{\input{I2toIp.pstex_t}}}
\caption{%
(a) An $\I_2$ instance. Each box denotes a facility, and the number inside the box is the 
number of co-located clients; a dashed arrow $i\rightarrow i'$ denotes that $i'$ is the
closest facility to $i$. \newline
(b) The corresponding $\I'$ instance. The boxes and circles represent supply points and
demand points respectively, and points inside a dotted oval are co-located. A solid box
denotes an uncapacitated supply point, and a dashed box denotes a capacitated facility
whose capacity is shown inside the box. The number inside a circle is the demand of that
demand point. The arrows indicate a solution $S$ to $\I'$, where $i$ and $i'$ are the two 
open uncapacitated supply points.}  
\label{I2toIp}
\end{figure}

\begin{lemma}[~\cite{Svitkina08}] \label{cdufl-cost}
There exists a solution to $\I'$ with facility cost $F\leq\dt C^*_{\I_2}$ and assignment
cost $C\leq C^*_{\I_2}$.
\end{lemma}

\begin{theorem} \label{cdufl-thm}
(i) Given any \cdufl instance, one can efficiently compute a solution with 
facility-opening cost $\hF\leq F^{\sol}+2C^{\sol}$ and assignment cost 
$\hC\leq F^{\sol}+C^{\sol}$, where $F^{\sol}$ and $C^{\sol}$ are the facility and
assignment costs of an arbitrary solution to the \cdufl instance.

\noindent
(ii) Thus, Lemma~\ref{cdufl-cost} implies that one can compute a solution to $\I'$ with
facility cost $F_{\I'}$ and assignment cost $C_{\I'}$ satisfying 
$F_{\I'}\leq(2+\dt)C^*_{\I_2},\ C_{\I'}\leq(1+\dt)C^*_{\I_2}$. 
\end{theorem}

\noindent
We defer the description of the local-search algorithm for \cdufl, and the proof of
Theorem~\ref{cdufl-thm} to Section~\ref{cdufl-ls}.
We first describe how to convert an 
$\I'$-solution to a solution to $\I_2$ with a small increase in cost, and show how this  
combined with Theorem~\ref{cdufl-thm} leads to the approximation bound for $\I_2$ stated
in Theorem~\ref{i2thm}. 

\subsection{Mapping an $\I'$-solution to an $\I_2$-solution} \label{i2-soln}
An $\I'$-solution need not directly translate to an $\I_2$ solution because an open
supply point $i$ may not supply (and hence, transfer) exactly $n_i$ units of demand (see,
e.g., $i$ and $i'$ in Fig.~\ref{I2toIp}(b)). Since we have uncapacitated supply points, we
have to consider both the cases where $i$ supplies more than $n_i$ demand (a situation not
encountered in~\cite{Svitkina08}), and less than $n_i$ demand. 
Suppose that we are given a solution $S$ to $\I'$ with facility cost $F^S$ and assignment
cost $C^S$ (see Fig.~\ref{I2toIp}(b)). Again, we abuse notation and use $F^S$ to also
denote the set of supply points that are opened in $S$. 
Let $N_i$ initialized to $n_i$ keep track of the number of clients at location
$i\in\F'$. Our goal is to reassign clients (using $S$ as a template) so that at the end
we have $N_i=0$ or $N_i\geq M$ for each $i\in\F'$.
Observe that once we have determined which facilities in $\F'$ will have $N_i\geq M$  
(i.e., the facilities to open in the $\I_2$-solution), 
one can find the best way of (re)assigning clients by solving a min-cost flow
problem. However, for purposes of analysis, it will often be convenient to explicitly
specify a (possibly suboptimal) reassignment. 
We may assume that:
(i) $\F^c\sse F^S$; 
(ii) if $S$ opens an uncapacitated supply point located at some $i\in\F'$ with $n_i>M$,
then the demand assigned to the capacitated supply point at $i$ equals its capacity
$n_i-M$; 
(iii) for each $i\in\F'$ with $n_i\leq M$, if the supply point at $i$ is open then it
serves the entire demand of the co-located demand point; and
(iv) at most one {\em uncapacitated} supply point serves, maybe partially, the demand
of any demand point; we say that this uncapacitated supply point satisfies the demand
point.  
We reassign clients in three phases.

\begin{labellist}{A}
\item {\bf (Removing capacitated supply points)\ } 
Consider any location $i\in\F'$ with $n_i>M$. Let $i^1$ and $i^2$ denote respectively the
capacitated and uncapacitated supply points located at $i$. If $i^1$ supplies $x$ units to
the demand  point at location $i'$, we transfer $x$ clients from location $i$ to $i'$. 
Now if $i^1$ has $y>0$ leftover units of capacity in $S$, then we ``move'' $y$ clients to
$i^2$ (which is not open in $S$). 
We update the $N_i$s accordingly. Note that this reassignment effectively gets rid of all
capacitated supply points.  
Thus, there is now exactly one uncapacitated supply point and at most one demand point at
each location $i\in\F'$; we refer to these simply as supply point $i$ and demand point $i$
below.  

\hspace*{-5ex}
\parbox{\textwidth}{
Let $X_i$ be the total demand from other locations assigned to supply point $i$. 
Let $\F^G=\{i\in\F': N_i<X_i\}$,  
$\F^R=\{i\in\F': N_i\geq X_i>0\}$, and $\F^B=\{i\in\F': X_i=0\}$. which is the set of
supply points that are not opened in $S$. 
Note that $N_i\geq\min\{n_i,M\}\geq\al M$ for all $i\in\F'$, and $N_i=\min\{n_i,M\}$ for
all $i\in\F^R\cup\F^G$ (because of properties (ii) and (iii) above).}

\item {\bf (Taking care of $\F^R$ and demand points satisfied by $\F^R$)\ }
For each $i\in\F^R$, if $i$ supplies $x$ units to demand point $i'$, we move $x$ clients
from $i$ to $i'$, and update $N_i, N_{i'}$.  
We now have $N_i=\min\{n_i,M\}-X_i$ residual clients at each $i\in\F^R$, which we must
reduce to 0, or increase to at least $M$. 
We follow the same procedure as in~\cite{Svitkina08}, which we sketch below. 

For each $i\in\F^R$, we include an edge $(i,i')$ where $i'\in\F'$ is  
the facility nearest to $i$ (recall that $c_{ii'}=l(i)$). We use an arbitrary but fixed
tie-breaking rule here, so each component of the resulting digraph is a directed tree
rooted at either (i) a node $r\in\F'\sm\F^R$, or (ii) a 2-cycle $(r,r'), (r',r)$, where
$r,r'\in\F^R$.  
We break up each component $\Gm$ into a collection of smaller components as follows. 
Essentially, we move the residual clients of supply points in the component bottom-up from
the leaves up to the root, cut off the component at the first 
node $u$ that accumulates at least $M$ clients, and recurse on the portion of the
component not containing $u$. More precisely, let $\Gm_u$ denote the subtree of $\Gm$
rooted at node $u\in\Gm$ (if $u$ belongs to a 2-cycle then we do not include the
other node of this 2-cycle in $\Gm_u$).
\begin{list}{--}{\itemsep=0ex \topsep=0ex \addtolength{\leftmargin}{-2ex}}
\item If $\sum_{i\in\Gm}N_i<M$, or if $\Gm$ is of type (i) and all children $u$ of the root 
satisfy $\sum_{i\in\Gm_u}N_i<M$, we leave $\Gm$ unchanged.  

\item Otherwise, let $u$ be a deepest (i.e., furthest from root) node in $\Gm$ such that 
$\sum_{i\in\Gm_u}N_i\geq M$. 
We delete the arc leaving $u$. 
If this disconnects $u$ from $\Gm\sm\Gm_u$, then we recurse on $\Gm\sm\Gm_u$.

\item Otherwise $u$ must belong to the root 2-cycle of $\Gm$. Let $r'$ be the other node
of this 2-cycle. 
If $\sum_{i\in\Gm_{r'}}N_i\geq M$, we delete $r'$'s outgoing arc (thus splitting $\Gm$
into $\Gm_u$ and $\Gm_{r'}$).
\end{list}
After applying the above procedure (to all components), if we are left with a component of
type (ii) with $\sum_{i\in\text{ component}}N_i\geq M$, we convert it to type (i) by 
arbitrarily deleting one of the arcs of the 2-cycle.
Thus, at the end of this process, we have two types of components.

\vspace{-1.5ex}
\begin{list}{(\alph{enumii})}{\usecounter{enumii} \itemsep=0ex
\addtolength{\leftmargin}{-1ex}} 
\item A tree $T$ rooted at a node $r$: we move the $N_i$ residual clients of each
non-root node $i\in T$ to $r$.
\item A type-(ii) tree $T$ with root $\{r,r'\}$: we must have $\sum_{i\in T}N_i<M$.  
Let $i'\in\F^B$ be the location nearest to $\{r,r'\}$; we move the $N_i$ residual clients 
of each $i\in T$ to $i'$.
\end{list}

\vspace{-1ex}
Update the $N_i$s to reflect the above reassignment. 
Observe that we now have $N_i=0$ or $N_i\geq M$ for each $i\in\F^R$, and 
each $i\in\F^B$ has $n_i\geq M$, or
is a demand point satisfied by a supply point in $\F^G$.
Figure~\ref{transition}(a) shows a snapshot after steps A1 and A2 have been executed on
the solution shown in Fig.~\ref{I2toIp}(b). Here $i'\in\F^R$ has one client left after  
moving clients to the bottom two facilities, which is then transferred to $i_3$.  
  
\item {\bf (Taking care of $\F^G$ and demand points satisfied by $\F^G$)\ }
For $i\in\F^G$, let $D(i)$ be the set of demand points $j\in\F',\ j\neq i$ 
satisfied by $i$, and let $D'(i)=\{j\in D(i): N_j<M\}$. 
Note that $D(i)\sse\F^B$. 
Phase A2 may only increase $N_j$ for all $j$ in $\F^B\cup\F^G$, so $N_j\geq\al M$ for 
all $j\in\F^G\cup\bigl(\bigcup_{i\in\F^G}D(i)\bigr)$.  

Fix $i\in\F^G$. We reassign clients so that $N_j=0$ or $N_j\geq M$ for all
$j\in\{i\}\cup D'(i)$, without decreasing $N_j$ for $j\in D(i)\sm D'(i)$. 
Applying this procedure to all supply points in $\F^G$ will complete our task. 
Define $Y_j=M-N_j$ (which is at most $M-n_j$) for $j\in D'(i)$. We consider two cases. 

\vspace{-1.5ex}
\begin{list}{--}{\itemsep=0ex \topsep=0ex \addtolength{\leftmargin}{-2ex}}
\item 
$\sum_{j\in D'(i)}Y_j\leq N_i$. For each $j\in D'(i)$, if $i$ supplies $x$
units to $j$, we transfer $x$ clients from $i$ to $j$. If $i$ is now left with less than
$M$ residual clients, we move these residual clients to the location in $D(i)$ nearest to
$i$.   

\item
$\sum_{j\in D'(i)}Y_j>N_i$ (see Fig.~\ref{transition}). Let $i_0=i$, and
$D'(i)=\{i_1,\ldots,i_t\}$, where $c_{i_1i}\leq \ldots\leq c_{i_ti}$. 
Let $\ell=t-\floor{\frac{\sum_{r=0}^t N_{i_r}}{M}} 
=\ceil{\frac{\sum_{r=1}^t Y_{i_r}-N_{i_0}}{M}}$, so $\ell\geq 1$ (and $\ell<t$ since 
$N_{i_0}+N_{i_1}\geq M$). Note that $\ell$ is the unique index such that 
$\sum_{r=\ell+1}^t Y_{i_r}\leq\sum_{r=0}^\ell N_{i_r}<\sum_{r=\ell+1}^t Y_{i_r}+M$. 
This enables us to transfer $Y_{i_q}$ clients to each $i_q,\ q=\ell+1,\ldots,t$ from
the locations $i_\ell,\ldots,i_0$---we do this by transferring all clients of $i_r$ 
(where $1\leq r\leq \ell$) before considering $i_{r-1}$---and be left with at most $M$
residual clients in $\{i_0,\ldots,i_\ell\}$. 
We argue that these residual clients are all concentrated at $i_0$ and $i_1$, with $i_1$
having at most $(1-\al)M$ residual clients. We transfer these residual clients to
$i_{\ell+1}$.
\end{list} 
\end{labellist}

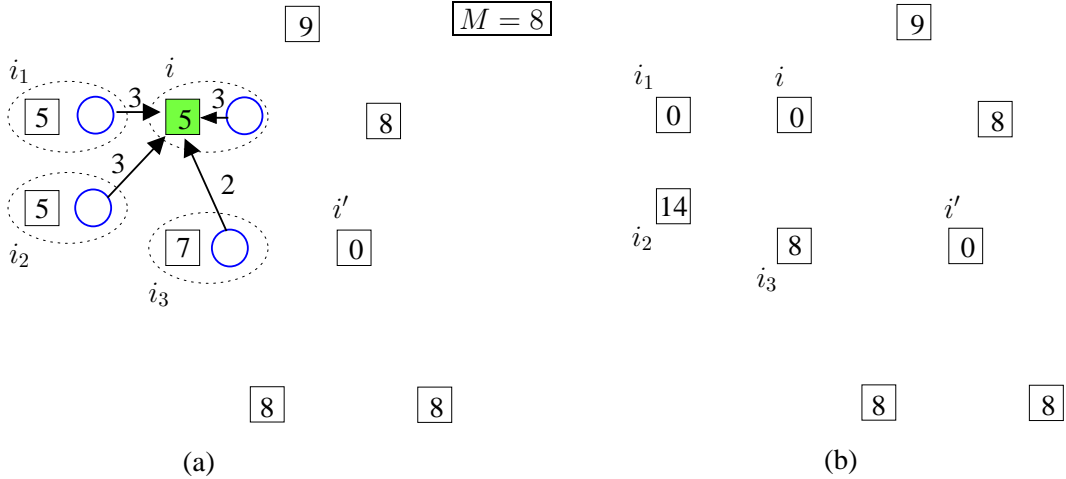
\begin{figure}[ht!]
\centerline{\resizebox{!}{2.5in}{\input{transition.pstex_t}}}
\caption{The number inside a box is the current value of $N_i$; the number labeling an
arrow is the demand assignment of the $\I'$-solution. The circles indicate demand points
$j$ with $N_j<M$. 
(a) The situation after running steps A1 and A2 on the solution in Fig.~\ref{I2toIp}(b). 
(b) The situation after running step A3.}
\label{transition}
\end{figure}

\begin{theorem} \label{i2map}
The above algorithm returns an $\I_2$-solution of cost at most
$\frac{F^S}{\dt\al}+C^S\bigl(\frac{1}{\al}+\frac{2\al}{2\al-1}\bigr)$. Thus, taking $S$ to
be the solution mentioned in part (ii) of Theorem~\ref{cdufl-thm}, and 
$\dt=\sqrt{\frac{2/\al}{1/\al+(2\al)/(2\al-1)}}$, 
we obtain a solution to $\I_2(\al)$ 
satisfying the approximation bound stated in Theorem~\ref{i2thm}.
\end{theorem}

\begin{proof}
Let $S_2$ denote the solution computed for $\I_2$.
For a supply point $i$ opened in $S$, we use $C^S_i$ to denote the cost incurred in
supplying demand from $i$ to the demand points satisfied by $i$; 
so $C^S=\sum_{i\in F^S}C^S_i$. 
At various steps, we transfer clients between locations according to the assignment
in the \cdufl solution $S$, and the cost incurred in this reassignment can be charged
against the $C^S_i$s of the appropriate supply points. 
So the cost of phase A1 is $\sum_{i\in\F^c}C^S_i$, and the cost of the first step of phase
A2 is $\sum_{i\in\F^R}C^S_i$. 

As in~\cite{Svitkina08}, we can bound the remaining cost of phase A2, incurred in
transferring clients according to the tree edges by 
$F^S/\dt\al+\bigl(\sum_{i\in\F^R}C^S_i\bigr)/(2\al-1)$.
When we move clients up to the root of a component, we move strictly less than $M$ clients
along any edge $(i,i')$ in that component, 
and since $i\in\F^R$, we pay at least $\dt\al M l(i)$ opening cost for $i$.
The only unaccounted cost now is the cost incurred in step (b) of phase A2, where we have
a tree $T$ rooted at $\{r,r'\}$. Let $i'\in\F^B$ be the location nearest to $\{r,r'\}$, and
(say) $c_{i'r}\leq c_{i'r'}$. Note that we have already bounded the cost in transferring
clients to $r$, so we only need to bound the cost incurred in transferring at most $M$
clients from $r$ to $i'$. 
This is at most 
$M\cdot\frac{C^S_r+C^S_{r'}}{X_r+X_{r'}}\leq\bigl(C^S_r+C^S_{r'}\bigr)/(2\al-1)$, because
$\{r,r'\}$ send $X_r+X_{r'}=(n_r+n_{r'})-(N_r+N_{r'})\geq(2\al-1)M$ units to demand points 
in $\F^B$, all of which are at distance at least $c_{i'r}$ from $\{r,r'\}$.

Finally, consider phase A3 and some $i\in\F^G$. If $\sum_{j\in D'(i)}Y_j\leq N_i$, then
the cost incurred is at most 
$C^S_i+M\cdot\frac{C^S_i}{X_i}\leq C^S_i\bigl(1+\frac{1}{\al}\bigr)$ 
(as $X_i>N_i\geq\al M$).  
Now consider the case $\sum_{j\in D'(i)}Y_j>N_i$. For any
$i_q\in\{i_{\ell+1},\ldots,i_t\}$ and any $i_r\in\{i_0,\ldots,i_\ell\}$, we have
$c_{i_ri_q}\leq 2c_{ii_q}$, so the cost of transferring $Y_{i_q}\leq M-n_{i_q}$ clients to
each $i_q,\ q=\ell+1,\ldots,t$ is at most $2C^S_i$. Observe that 
$(t-\ell+1)M>\sum_{r=0}^t N_{i_r}$, i.e., 
$M+\sum_{q=\ell+1}^t Y_{i_r}>\sum_{r=0}^\ell N_{i_r}$, so after this reassignment, there
are less than $M$ residual clients in $i_0,\ldots,i_\ell$. By our order of transferring
clients, all these residual clients are at $i_0, i_1$ (otherwise we would have at least
$N_{i_0}+N_{i_1}\geq M$ residual clients) with at most $M-N_{i_0}\leq (1-\al)M$ of them
located at $i_1$. 
The cost of reassigning these residual clients is at most
$(1-\al)Mc_{ii_1}+Mc_{ii_{\ell+1}}\leq
(1-\al)M\cdot\frac{C^S_i}{\sum_{r=1}^t Y_{i_r}}+
M\cdot\frac{C^S_i}{\sum_{r=\ell+1}^t Y_{i_r}}$, since 
$C^S_i$ is the total cost of supplying at least $Y_{i_r}$ demand to each
$i_r,\ r=1,\ldots, t$. 
The latter expression is at most 
$C^S_i\bigl(\frac{1-\al}{\al}+\frac{1}{2\al-1}\bigr)$, since
$\sum_{r=1}^t Y_{i_r}>N_{i_0}\geq\al M$, 
$\sum_{r=\ell+1}^t Y_{i_r}>\sum_{r=0}^\ell N_{i_r}-M\geq(2\al-1)M$.)
Thus, the cost of $S_2$ is at most
$$ 
\frac{F^S}{\dt\al}+\sum_{i\in\F^c}C^S_i+\sum_{i\in\F^R}C^S_i\cdot\Bigl(1+\tfrac{1}{2\al-1}\Bigr)
+\sum_{i\in\F^G}C^S_i\cdot\max\Bigl\{1+\tfrac{1}{\al},2+\tfrac{1-\al}{\al}+\tfrac{1}{2\al-1}\Bigr\}
\leq\frac{F^S}{\dt\al}+C^S\Bigl(\tfrac{1}{\al}+\tfrac{2\al}{2\al-1}\Bigr). 
$$
So if $S$ is the solution given by part (ii) of Theorem~\ref{cdufl-thm}, the cost of $S_2$
is at most
$\bigl(\frac{2}{\dt\al}+\frac{1}{\al}+(1+\dt)(\frac{1}{\al}+\frac{2\al}{2\al-1})\bigr)
C^*_{\I_2}$, and plugging in the value of $\dt$ yields the 
$g(\al)=\frac{2}{\al}+\frac{2\al}{2\al-1}+2\sqrt{\frac{2}{\al^2}+\frac{4}{2\al-1}}$
approximation bound stated in Theorem~\ref{i2thm}.
\end{proof}

\subsection{A local-search based approximation algorithm for \cdufl} \label{cdufl-ls} 
We now describe our local-search algorithm for \cdufl, which leads to the proof of
Theorem~\ref{cdufl-thm}. 
Let $\hcF=\hcFu\cup\hcFc$ be the facility-set of the \cdufl
instance, where $\hcFu\cap\hcFc=\es$. Here, $\hcFu$ are the uncapacitated facilities
with opening costs $\{\htf_i\}$, and facilities in $\hcFc$ have (finite) capacities
$\{u_i\}$ and zero opening costs. Let $\hcD$ be the set of clients and $\hc_{ij}$ be the
cost of assigning client $j$ to facility $i$. 
The goal is to open facilities and assign clients to open facilities (respecting the
capacities) so as to minimize the sum of the facility-opening and client-assignment costs.   
We can find the best assignment of clients to open facilities by solving a network
flow problem, so we focus on determining the set of facilities to open.   

The local-search algorithm consists of three moves: $\add(i')$, $\drop(i)$, 
$\swap(i,i')$, which respectively, add a facility $i'$ not currently open, delete a
facility $i$ that is currently open, and swap facility $i$ that is open with facility $i'$
that is not open.  
We note that {\em all} previous (local-search) algorithms for \cfl that work with
non-uniform capacities use moves that are more complicated than the moves above (and
involve adding and/or deleting multiple facilities at a time).  
The algorithm repeatedly executes the best cost-improving move (if one exists) until no
such move exists. 
(As mentioned earlier, to ensure polynomial time, we only consider moves that
yield significant improvement and hence terminate at an approximate local optimum; but
this has only a marginal effect on the approximation bound.)  
We assume for simplicity that each client has unit demand. This is
without loss of generality because, even with non-unit client-demands, one can compute the
best local-search move (and hence run the algorithm), and for the purposes of analysis,
one can always treat a client with integer demand $d$ as $d$ co-located unit-demand
clients.  

\vspace{-1ex}
\paragraph{Analysis.}
Let $\hS$ denote a local-optimum returned by the algorithm, with facility-opening cost
(and set of open facilities) $\hF$ and assignment cost $\hC$. 
Let $\sol$ be an arbitrary \cdufl solution, with facility-cost (and set of open
facilities) $F^{\sol}$ and assignment cost $C^{\sol}$.
Note that we may assume that $\hcFc\sse\hF\cap F^{\sol}$.  
For a facility $i$, we use $\hcDS(i)$ and $\hcDsol(i)$ to denote respectively the
(possibly empty) set of clients served by $i$ in $\hS$ and $\sol$. For a client $j$, let
$\hC_j$ and $C^{\sol}_j$ be the assignment cost of $j$ in $\hS$ and $\sol$ respectively.

We borrow ideas from the analysis of the corresponding local-search algorithm
for \ufl in~\cite{AryaGKMMP01}, 
but the presence of capacitated facilities means that we need to reassign clients more
carefully to analyze the change in assignment cost due to a local-search move. In
particular, unlike the analysis in~\cite{AryaGKMMP01}, where upon deletion of a facility 
$s\in\hF$ 
we reassign only the clients currently assigned to $s$, in our case (as in the
analysis of local-search algorithms for \cfl), we need to perform a more ``global''
reassignment (i.e., even clients not assigned to $s$ may get reassigned) along certain
(possibly long) paths in a suitable graph. This also means that we need to construct 
a suitable mapping between paths instead of the client-mapping considered
in~\cite{AryaGKMMP01}. 

We construct a directed graph $G$ with node-set $\hcD\cup\hcF$, and arcs from $i$ 
to all clients in $\hcDS(i)$ and arcs from all clients in $\hcDsol(i)$ to $i$, for every
facility $i$.
Via standard flow-decomposition, we can decompose $G$ into a collection of (simple) 
paths $\Pc$, and cycles $\Rc$, so that (i) each facility $i$ appears as the starting point
of $\max\{0,|\hcDS(i)|-|\hcDsol(i)|\}$ paths, and the ending point of
$\max\{0,|\hcDsol(i)|-|\hcDS(i)|\}$ paths, and (ii) each client $j$ appears on a unique
path $P_j$ or on a cycle.  
Let $\Pcst{s}\sse\Pc$ and $\Pcend{o}\sse\Pc$ denote
respectively the collection of paths starting at $s$ and ending at $o$,
and $\Pc(s,o)=\Pcst{s}\cap\Pcend{o}$.
For a path $P=\{i_0,j_0,i_1,j_1,\ldots,i_k,j_k,i_{k+1}:=o\}\in\Pc$, define
$\hcD(P)=\{j_0,\ldots,j_k\}$, $\head(P)=j_0$, and $\tail(P)=j_k$. 
A {\em shift} along $P$ means that
we reassign client $j_r$ to $i_{r+1}$ for each $r=0,\ldots,k$ (opening $o$ if
necessary). Note that this is feasible, since if $o\in\hcFc$, we know that 
$|\hcDS(o)|\leq |\hcDsol(o)|-1\leq u_o-1$. Let
$\shift(P):=\sum_{j\in\hcD(P)}\bigl(C^{\sol}_j-\hC_j\bigr)$ be the increase in assignment
cost due to this reassignment, which is an upper bound on the actual increase in
assignment cost if $o$ is added to $\hF$. Also, let 
$\cost(P):=\sum_{j\in\hcD(P)}\bigl(C^\sol_j+\hC_j\bigr)$.
We define a shift along a cycle $R\in\Rc$ similarly, letting
$\shift(R):=\sum_{j\in\hcD\cap R}\bigl(C^\sol_j-\hC_j\bigr)$.
By considering a shift operation for every path
and cycle in $\Pc\cup\Rc$ (i.e., suitable $\add$ moves), we get the 
following result. 

\begin{lemma} \label{asgncost}
For every $o\in F^{\sol}$ and any $\Qc\sse\Pcend{o}$, we have 
$\sum_{P\in\Qc}\shift(P)\geq 
\begin{cases}-\htf_o& \text{if $o\in F^{\sol}\sm\hF$}, \\ 
0& \text{otherwise}.\end{cases}$
For every cycle $R\in\Rc$, we have $\shift(R)\geq 0$.
Thus, we have $\hC\leq F^{\sol}+C^{\sol}$.  
\end{lemma}

\paragraph{Bounding the opening cost of facilities in \boldmath $\hF\sm F^{\sol}$.} For
this, we only need paths that start at facility in $\hF\sm F^{\sol}$. Note that all
facilities in $(\hF\sm F^{\sol})\cup(F^{\sol}\sm\hF)$ are  
{\em uncapacitated}. To avoid excessive notation, 
for a facility $o\in F^{\sol}\sm\hF$, we now use $\Pcend{o}$ to refer to the collection of
paths ending in $o$ that start in $\hF\sm F^{\sol}$. 
(As before, $\Pc(s,o)$ is the set of paths that start at $s$ and end at $o$.) 
For any $o\in F^{\sol}\sm\hF$, we can obtain a 1-1 mapping $\pi:\Pcend{o}\mapsto\Pcend{o}$
such that if $P\in\Pc(s,o),\ s\in\hF\sm F^{\sol}$ and $\pi(P)=P'\in\Pc(s',o)$, then 
(i) if $|\Pc(s,o)|\leq\frac{|\Pcend{o}|}{2}$, we have $s\neq s'$; 
(ii) if $s=s'$, then $P=P'$; and
(iii) $\pi(P')=P$.
Say that $o\in F^{\sol}\sm\hF$ is {\em captured} by $s$ if
$|\Pc(s,o)|>\frac{|\Pcend{o}|}{2}$. Note that $o$ is captured by at most one facility in
$\hF$. Call a facility in $\hF\sm F^{\sol}$ {\em good} 
if it does not capture any facility, and {\em bad} otherwise. 

\begin{lemma} \label{goodf}
For any good facility $s$, we have 
\begin{equation}
\htf_s\leq\sum_{P\in\Pcst{s}}\shift(P)+
\sum_{o\notin\hF,P\in\Pc(s,o)}\cost\bigl(\pi(P)\bigr). \label{goodf-ineq}
\end{equation} 
\end{lemma}

\begin{proof}
Consider the move $\drop(s)$. We upper bound the increase in
reassignment cost as follows. Consider $j\in\hcDS(s)$, and let $P_j\in\Pc(s,o)$. (Recall
that $P_j$ is the unique path containing $j$.)
If $o\in\hF\cap F^{\sol}$, then we perform a shift along $P_j$. Otherwise, let
$\pi(P_j)\in\Pc(s',o)$, where $s'\neq s$. We reassign all clients on $P_j$ except
$\tail(P_j)$ as in the shift operation, and reassign $\tail(P_j)$ to $s'$.
Let $k=\tail(P_j)$. 
Since $c_{s'k}\leq c_{s'o}+C^{\sol}_{k}\leq\cost\bigl(\pi(P_j)\bigr)+C^{\sol}_{k}$, 
the increase in cost by reassigning clients on $P_j$ this way is at most
$\cost\bigl(\pi(P_j)\bigr)+C^{\sol}_{k}-\hC_{k}+
\sum_{j'\in\hcD(P_j)\sm\{k\}}\bigl(C^{\sol}_{j'}-\hC_{j'}\bigr)$.
Thus, the actual increase in cost due to this move, which should be nonnegative,  
is at most
$$
-\htf_s+\sum_{o\in\hF, P\in\Pc(s,o)}\shift(P)+
\sum_{o\notin\hF,P\in\Pc(s,o)}\Bigl[\shift(P)+\cost\bigl(\pi(P)\bigr)\Bigr].
\\[-20pt]
$$ 
\end{proof}

Now consider a bad facility $s$. 
Let $\capt_s\sse F^{\sol}\sm\hF$ be the facilities captured by $s$, and let
$o_s\in\capt_s$ be the facility nearest to $s$.  

\begin{lemma} \label{badf}
For any bad facility $s$, we have
\begin{equation}
\htf_s\leq\sum_{o\in\capt_s}\htf_o+\sum_{P\in\Pcst{s}}\shift(P)+
\sum_{\substack{o\notin\hF \\ P\in\Pc(s,o):\pi(P)\neq P}}\cost\bigl(\pi(P)\bigr)+
\sum_{\substack{o\in\capt_s\sm\{o_s\}\\ P\in\Pc(s,o):\pi(P)=P}}\cost(P).
\label{badf-ineq}
\end{equation}
\end{lemma}

\begin{proof}
Consider the move $\swap(s,o_s)$. We reassign client $j\in\hcDS(s)$ as follows. Let
$P_j\in\Pc(s,o)$. 
\begin{list}{$\bullet$}{\usecounter{enumi} \itemsep=0ex \addtolength{\leftmargin}{-2ex}}
\item If $o\in\hF\cap F^{\sol}$, or $o=o_s$ and $\pi(P_j)=P_j$, we perform a
shift along $P_j$. The increase in assignment cost is at most $\shift(P_j)$.

Otherwise, let $\pi(P_j)\in\Pc(s',o)$. 

\item If $\pi(P_j)\neq P_j$ (so $s'\neq s$), we reassign $\hcD(P_j)\sm\{\tail(P_j)\}$ as
in the shift operation, and assign $\tail(P_j)$ to $s'$. 
As in the proof of Lemma~\ref{goodf}, the increase in assignment cost is at most
$\shift(P_j)+\cost\bigl(\pi(P_j)\bigr)$. 

\item If $\pi(P_j)=P_j$ (so $o\neq o_s$), we assign $j$ to $o_s$. Note that
$c_{o_sj}\leq\hC_j+c_{so_s}\leq\hC_j+c_{so}\leq\hC_j+\cost(P_j)$, so the increase in
assignment cost is at most $\cost(P_j)$.
\end{list}
This gives the inequality
\begin{equation}
\begin{split}
0 \leq \htf_{o_s}-\htf_s & +
\sum_{\substack{P\in\Pc(s,o): o\in\hF\text{ or}\\ o=o_s,\ \pi(P)=P}}\shift(P)+
\sum_{o\notin\hF}\sum_{P\in\Pc(s,o): \pi(P)\neq P}\Bigl[\shift(P)+\cost\bigl(\pi(P)\bigr)\Bigr] \\
& + \sum_{o\notin\hF: o\neq o_s}\sum_{P\in\Pc(s,o): \pi(P)=P}\cost(P). \label{swap-ineq}
\end{split}
\end{equation}
Now consider the operation $\add(o)$ for all $o\in\capt_s\sm\{o_s\}$, and 
apply Lemma~\ref{asgncost} taking $\Qc=\{P\in\Pc(s,o): \pi(P)=P\}$.
This yields the inequality
$0\leq \htf_o+\sum_{P\in\Pc(s,o):\pi(P)=P}\shift(P)$ for each $o\in\capt(s)\sm\{o_s\}$.
Adding these inequalities to \eqref{swap-ineq}, and rearranging proves the lemma.
\end{proof}

\begin{proofof}{Theorem~\ref{cdufl-thm}}
We prove part (i); part (ii) follows directly from part (i) and Lemma~\ref{cdufl-cost}.
Lemma~\ref{asgncost} bounds $\hC$. 
Consider adding \eqref{goodf-ineq} for all good facilities and \eqref{badf-ineq} for all
bad facilities, and the vacuous equality $\htf_i=\htf_i$ for all $i\in\hF\cap F^{\sol}$.
The LHS of the resulting inequality is precisely $\hF$. 
The $\htf_i$s on the RHS add up to give at most $F^{\sol}$. 
We claim that each path $P\in\bigcup_{s\in\hF\sm F^{\sol}}\Pcst{s}$ contributes at most
$\shift(P)+\cost(P)=2\sum_{j\in\hcD(P)}C^{\sol}_j$ to the RHS. Thus the RHS is at most
$F^{\sol}+2C^{\sol}$, and we obtain that $\hF\leq F^{\sol}+2C^{\sol}$.

Each path $P$ in $\bigcup_{s\notin F^{\sol},o\in\hF}\Pc(s,o)$ appears exactly once,
either in \eqref{goodf-ineq} or in \eqref{badf-ineq}, and contributes $\shift(P)$. 
Now consider a path $P\in\bigcup_{s\notin F^{\sol},o\notin\hF}\Pc(s,o)$, and let
$\pi(P)=P'\in\Pc(s',o)$. Note that $\pi(P')=P$. 
If $P'\neq P$, then $P$ appears twice in our inequality-system: once in the inequality
for $s$ contributing $\shift(P)$ (due to $P$), and once in the inequality for $s'$
contributing $\cost(P)$ (due to $P'$). 
If $P'=P$, then $s=s'$ and $s$ is a bad facility; now $P$ appears only in
\eqref{badf-ineq} (for $s$) and contributes either $\shift(P)$ if $o=o_s$, or
$\shift(P)+\cost(P)$ otherwise. 
\end{proofof}

\begin{labeledthm}{Corollary of Theorem~\ref{cdufl-thm}:}
There is a $\bigl(1+\sqrt{2}\bigr)$-approximation algorithm for \cdufl.
\end{labeledthm}

\begin{proof}
We take $\sol$ in part (i) of Theorem~\ref{cdufl-thm} to be an optimum solution (with cost
$F^{\opt}+C^{\opt}$) to the instance, and scale the facility costs by $\sg$ before running
the local-search algorithm. 
The solution returned has cost 
$F+C\leq \bigl(F^{\opt}+\frac{2}{\sg}\cdot C^{\opt}\bigr)+\bigl(\sg F^{\opt}+C^{\opt}\bigr)$.
Setting $\sg=\sqrt{2}$ yields the result.
\end{proof}

\appendix

\section{Integrality-gap example for the natural LP-relaxation for \cdufl}
\label{cduflgap} 
Let $\bigl(\hcF=\hcFu\cup\hcFc,\hcD,\{\htf_i\},\{u_i\},\{\hc_{ij}\}\bigr)$ be a \cdufl  
instance with facility-set $\hcF$ (where $u_i=\infty$ for all $i\in\hcFu$, and $\htf_i=0$
for all $i\in\hcFc$), and client-set $\hcD$.   
We consider the following LP-relaxation. We use $i$ to index facilities, and $j$ to index
clients. Note that we may assume that all facilities in $\hcFc$ are open.
\begin{alignat*}{3}
\min & \quad & \sum_{i\in\hcFu}\htf_iy_i & +\sum_{j,i}\hc_{ij}x_{ij} \tag{LP}
\label{cdufllp} \\ 
\text{s.t.} && \sum_i x_{ij} & \geq 1 \qquad && \frall j \\[-2ex]
&& x_{ij} & \leq y_i && \frall i\in\hcFu, j \\
&& \sum_j x_{ij} & \leq u_i && \frall i\in\hcFc \\[-1ex]
&& x_{ij},y_i & \geq 0 && \frall i, j.
\end{alignat*} 
Here $y_i$ denotes if facility $i$ is open, and $x_{ij}$ denotes if client $j$ is assigned
to facility $i$. (We assume that each client has unit demand.) 

Now consider the following simple \cdufl instance. We have two facilities $i$ and $i'$,
and $u+1$ clients, all present at the same location. 
Facility $i$ is uncapacitated and has opening cost $f$, and facility
$i'$ has capacity $u$ (and zero opening cost). 
Any solution to \cdufl must open facility $i$ and therefore incur cost at least
$f$. However, there is a feasible solution to \eqref{cdufllp} of cost $\frac{f}{u+1}$: we
set $y_i=\frac{1}{u+1}$, and $x_{ij}=\frac{1}{u+1},\ x_{i'j}=\frac{u}{u+1}$. Thus, the
integrality gap of \eqref{cdufllp} is at least $u+1$.

\section{The locality gap of a local-search algorithm for LBFL}
\label{app:LBFLlocgap}
We show that the local-search algorithm based on $\add$, $\drop$, and $\swap$ moves---that
is, adding/dropping one facility (with $\add$ permitted only if it preserves feasibility),
or deleting one facility and adding another---has a bad {\em locality gap}, which is the
maximum ratio between the cost of a locally-optimal solution and that of an (globally)
optimal solution. Consider the \lbfl instance shown below with 
facility-set $\F =\{o,s_1,s_2,\ldots,s_M\}$, and client-set 
$\D=\D_1\cup\D_2\cup\ldots\cup\D_M$, where the $\D_i$s are disjoint sets of size $M$. The 
facility-opening costs are as follows: $f_o = M^2+\epsilon$ and $f_{s_i} = M $ for each
$i\in\{1,2,\ldots,M\}$. For each $i=1,\ldots,m$ and each client $j\in\D_i$, we have
$c_{oj}=1,\ c_{s_ij}=M$. All other distances are defined by taking the metric completion
with respect to these $c_{ij}$s. 
One can verify that the solution $S$ which opens the facilities $\{s_1,s_2,\ldots,s_M\}$
is a local optimum. 
The cost of this solution is $M^2+M^3$. However, the optimal solution opens facility
$\{o\}$, and incurs a total cost of $2M^2+\epsilon$. Thus, the locality gap is at least
$M/2$.

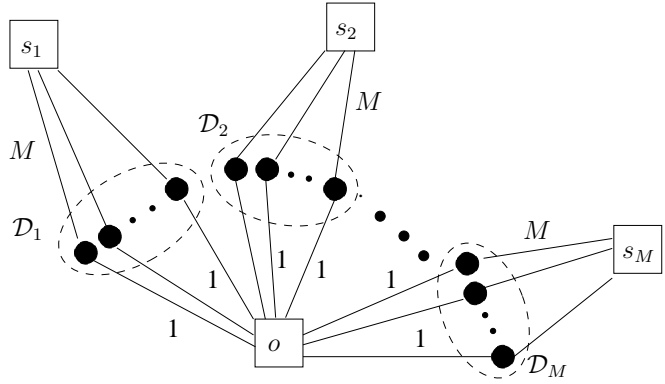
\begin{figure}[ht!]
\centerline{\resizebox{!}{2in}{\input{locgapfp.pstex_t}}}
\label{locgapfp}
\end{figure}

We can modify this example to show that the locality gap remains bad, even if
aim for a bicriteria solution and consider an $\add$ move to be permissible if every open 
facility can be assigned at least $\al M$ clients. The only change is that each set $\D_i$
now has $\al M$ clients: $S$ is still a local optimum, and the locality gap is therefore
at least $\al M/2$.

\paragraph{Bad example with zero facility-opening costs.}
Even in the setting where all facilities have zero opening cost (as in the $\I_2$
instance), we can construct bad examples for local-search based on $\add$, $\drop$, and 
$\swap$ moves. 
For simplicity, first suppose that $M=2$. 
Consider a cycle with $4k$ nodes, which are labeled 
$o_0, j_0, s_0, j_1, o_1, j_2, s_1, j_3, \ldots, o_r, j_{2r}, s_r, j_{2r+1}, \ldots, 
o_{k-1}, j_{2k-2}, s_{k-1}, j_{2k-1}, o_0$. 
We have $2k$ facilities $\F =\{o_0,\ldots,o_{k-1},s_0,\ldots,s_{k-1}\}$, and $2k$ clients 
$\D=\{j_0,j_1,\ldots,j_{2k-1}\}$ (see Fig.~\ref{locgapf0}). We define the following
distances. 
\begin {itemize}
\item $c_{o_i j_{2i\mod 2k}} = c_{o_i j_{(2i-1)\mod 2k}}=1$ for all $i=0,\ldots k-1$.
\item $c_{s_i j_{2i}} = c_{s_i j_{(2i+1)}}=k-\epsilon$ for all $i=0,\ldots,k-1$.
\end{itemize}
All other distances are defined by taking the metric completion with respect to these
$c_{ij}$s. 

\begin{figure}[ht!]
\centerline{\resizebox{!}{3in}{\input{locgapf0.pstex_t}}}
\caption{Bad locality-gap example with 0 facility costs}
\label{locgapf0}
\end{figure}
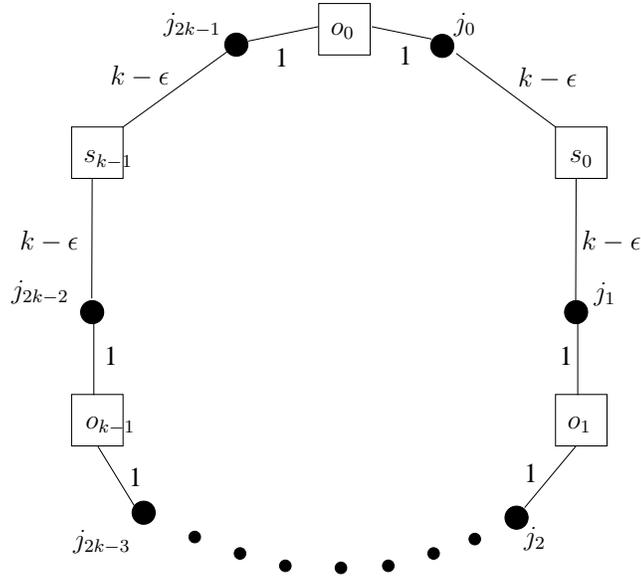

The solution $S$ which opens facilities $\{s_0,s_1,\ldots,s_{k-1}\}$ is a local optimum:
no $\add$ move is feasible, and it is easy to see that no $\drop$ move improves the
cost. Consider a swap move, which we may assume is of the form $\swap(s_r,o_0)$ by
symmetry. The new client-assignment will not necessarily assign the clients $j_{2r}$ and
$j_{2r+1}$ (which were previously assigned to $s_r$) to $o_0$. However, the intuition is
that the long cycle 
will lead to a large increase in assignment cost. 
The optimal way of reassigning clients is to assign $j_{2k-1}, j_0$ to $o_o$, assign
$j_{2i+1}, j_{2i+2}$ to $s_i$ for $i\in\{0,\ldots,r-1\}$ (which is empty if $r=0$), and
assign $j_{2i}, j_{2i-1}$ to $s_i$ for $i\in\{r+1,\ldots,k-1\}$ (which is empty if
$r=k-1$). The cost increase due to this reassignment is $2(1-k+\e)+(k-1)\cdot 2>0$. 
Thus, $S$ is a local optimum.

The cost of $S$ is $2k(k-\epsilon)$. However, the optimal solution opens facilities
$\{o_0,\ldots,o_{k-1}\}$, and has a total cost of $2k$. So this instance shows a locality
gap of $k$, and since $k$ can be made arbitrarily large, this shows an unbounded locality
gap.  

\medskip
The above example can be extended to all values of $M$. For each $M$, let
$G^M$ be an $M$-regular bipartite graph with vertex set 
$V=\{o_1,o_2,...,o_\ell\}\cup\{s_1,s_2,...,s_\ell\}$ with a large girth $T$. We use $G^M$
to construct the following \lbfl instance. 
The set of facilities is $\{o_1,\ldots,o_\ell,s_1,\ldots,s_\ell\}$. 
For each edge $(s_n,o_m)$ in $G^M$, we create a client $j_{nm}$ with 
$c_{s_n j_{nm}}=T-\epsilon$ and $c_{o_m j_{nm}}=1$. 
As before, one can argue that the solution $S$ that opens facilities
$\{s_1,s_2,\ldots,s_\ell\}$ is a local optimum.  
The cost of this solution is $\ell M(T-\epsilon)$, whereas the solution that opens
facilities  $\{o_1,\ldots,o_\ell\}$ has total cost of $\ell M $. So the locality gap is
$T$, which can be made arbitrarily large.     

\end{document}

%% file: I2toIp.pstex_t
\begin{picture}(0,0)%
\includegraphics{I2toIp.pstex}%
\end{picture}%
\setlength{\unitlength}{3947sp}%
\begingroup\makeatletter\ifx\SetFigFont\undefined%
\gdef\SetFigFont#1#2#3#4#5{%
  \reset@font\fontsize{#1}{#2pt}%
  \fontfamily{#3}\fontseries{#4}\fontshape{#5}%
  \selectfont}%
\fi\endgroup%
\begin{picture}(8988,4161)(589,-6668)
\put(796,-3616){\makebox(0,0)[b]{\smash{{\SetFigFont{14}{16.8}{\rmdefault}{\mddefault}{\updefault}{\color[rgb]{0,0,0}5}%
}}}}
\put(7036,-4681){\makebox(0,0)[lb]{\smash{{\SetFigFont{14}{16.8}{\rmdefault}{\mddefault}{\updefault}{\color[rgb]{0,0,0}2}%
}}}}
\put(9286,-5956){\makebox(0,0)[b]{\smash{{\SetFigFont{14}{16.8}{\rmdefault}{\mddefault}{\updefault}{\color[rgb]{0,0,0}1}%
}}}}
\put(7726,-5956){\makebox(0,0)[lb]{\smash{{\SetFigFont{14}{16.8}{\rmdefault}{\mddefault}{\updefault}{\color[rgb]{0,0,0}3}%
}}}}
\put(5929,-3623){\makebox(0,0)[lb]{\smash{{\SetFigFont{14}{16.8}{\rmdefault}{\mddefault}{\updefault}{\color[rgb]{0,0,0}3}%
}}}}
\put(3151,-3781){\makebox(0,0)[b]{\smash{{\SetFigFont{14}{16.8}{\rmdefault}{\mddefault}{\updefault}{\color[rgb]{0,0,0}6}%
}}}}
\put(8011,-2881){\makebox(0,0)[lb]{\smash{{\SetFigFont{14}{16.8}{\rmdefault}{\mddefault}{\updefault}{\color[rgb]{0,0,0}3}%
}}}}
\put(8657,-3657){\makebox(0,0)[lb]{\smash{{\SetFigFont{14}{16.8}{\rmdefault}{\mddefault}{\updefault}{\color[rgb]{0,0,0}2}%
}}}}
\put(8542,-4700){\makebox(0,0)[lb]{\smash{{\SetFigFont{14}{16.8}{\rmdefault}{\mddefault}{\updefault}{\color[rgb]{0,0,0}3}%
}}}}
\put(7153,-3623){\makebox(0,0)[lb]{\smash{{\SetFigFont{14}{16.8}{\rmdefault}{\mddefault}{\updefault}{\color[rgb]{0,0,0}3}%
}}}}
\put(1696,-3616){\makebox(0,0)[b]{\smash{{\SetFigFont{14}{16.8}{\rmdefault}{\mddefault}{\updefault}{\color[rgb]{0,0,0}5}%
}}}}
\put(766,-4336){\makebox(0,0)[b]{\smash{{\SetFigFont{14}{16.8}{\rmdefault}{\mddefault}{\updefault}{\color[rgb]{0,0,0}5}%
}}}}
\put(2926,-4696){\makebox(0,0)[b]{\smash{{\SetFigFont{14}{16.8}{\rmdefault}{\mddefault}{\updefault}{\color[rgb]{0,0,0}5}%
}}}}
\put(1681,-4681){\makebox(0,0)[b]{\smash{{\SetFigFont{14}{16.8}{\rmdefault}{\mddefault}{\updefault}{\color[rgb]{0,0,0}6}%
}}}}
\put(3646,-5941){\makebox(0,0)[b]{\smash{{\SetFigFont{14}{16.8}{\rmdefault}{\mddefault}{\updefault}{\color[rgb]{0,0,0}7}%
}}}}
\put(2221,-5956){\makebox(0,0)[b]{\smash{{\SetFigFont{14}{16.8}{\rmdefault}{\mddefault}{\updefault}{\color[rgb]{0,0,0}5}%
}}}}
\put(7894,-4292){\makebox(0,0)[lb]{\smash{{\SetFigFont{14}{16.8}{\rmdefault}{\mddefault}{\updefault}{\color[rgb]{0,0,0}$i'$}%
}}}}
\put(6544,-3227){\makebox(0,0)[lb]{\smash{{\SetFigFont{14}{16.8}{\rmdefault}{\mddefault}{\updefault}{\color[rgb]{0,0,0}$i$}%
}}}}
\put(3901,-2836){\makebox(0,0)[lb]{\smash{{\SetFigFont{14}{16.8}{\rmdefault}{\mddefault}{\updefault}{\color[rgb]{0,0,0}\fbox{$M=8$}}%
}}}}
\put(2596,-2952){\makebox(0,0)[b]{\smash{{\SetFigFont{14}{16.8}{\rmdefault}{\mddefault}{\updefault}{\color[rgb]{0,0,0}11}%
}}}}
\put(1966,-6586){\makebox(0,0)[lb]{\smash{{\SetFigFont{14}{16.8}{\rmdefault}{\mddefault}{\updefault}{\color[rgb]{0,0,0}(a) $\mathcal{I}_2$}%
}}}}
\put(5956,-6586){\makebox(0,0)[lb]{\smash{{\SetFigFont{14}{16.8}{\rmdefault}{\mddefault}{\updefault}{\color[rgb]{0,0,0}(b) $\I'$, and a solution $S$ for $\I'$}%
}}}}
\end{picture}%

%% file: transition.pstex_t
\begin{picture}(0,0)%
\includegraphics{transition.pstex}%
\end{picture}%
\setlength{\unitlength}{3947sp}%
\begingroup\makeatletter\ifx\SetFigFont\undefined%
\gdef\SetFigFont#1#2#3#4#5{%
  \reset@font\fontsize{#1}{#2pt}%
  \fontfamily{#3}\fontseries{#4}\fontshape{#5}%
  \selectfont}%
\fi\endgroup%
\begin{picture}(8535,3854)(309,-6484)
\put(8710,-5959){\makebox(0,0)[b]{\smash{{\SetFigFont{14}{16.8}{\rmdefault}{\mddefault}{\updefault}{\color[rgb]{0,0,0}8}%
}}}}
\put(3775,-5974){\makebox(0,0)[b]{\smash{{\SetFigFont{14}{16.8}{\rmdefault}{\mddefault}{\updefault}{\color[rgb]{0,0,0}8}%
}}}}
\put(3901,-2836){\makebox(0,0)[lb]{\smash{{\SetFigFont{14}{16.8}{\rmdefault}{\mddefault}{\updefault}{\color[rgb]{0,0,0}\fbox{$M=8$}}%
}}}}
\put(7342,-5956){\makebox(0,0)[b]{\smash{{\SetFigFont{14}{16.8}{\rmdefault}{\mddefault}{\updefault}{\color[rgb]{0,0,0}8}%
}}}}
\put(8062,-4681){\makebox(0,0)[b]{\smash{{\SetFigFont{14}{16.8}{\rmdefault}{\mddefault}{\updefault}{\color[rgb]{0,0,0}0}%
}}}}
\put(7657,-2881){\makebox(0,0)[b]{\smash{{\SetFigFont{14}{16.8}{\rmdefault}{\mddefault}{\updefault}{\color[rgb]{0,0,0}9}%
}}}}
\put(8302,-3661){\makebox(0,0)[b]{\smash{{\SetFigFont{14}{16.8}{\rmdefault}{\mddefault}{\updefault}{\color[rgb]{0,0,0}8}%
}}}}
\put(6682,-3631){\makebox(0,0)[b]{\smash{{\SetFigFont{14}{16.8}{\rmdefault}{\mddefault}{\updefault}{\color[rgb]{0,0,0}0}%
}}}}
\put(5686,-3616){\makebox(0,0)[b]{\smash{{\SetFigFont{14}{16.8}{\rmdefault}{\mddefault}{\updefault}{\color[rgb]{0,0,0}0}%
}}}}
\put(5686,-4351){\makebox(0,0)[b]{\smash{{\SetFigFont{14}{16.8}{\rmdefault}{\mddefault}{\updefault}{\color[rgb]{0,0,0}14}%
}}}}
\put(6667,-4666){\makebox(0,0)[b]{\smash{{\SetFigFont{14}{16.8}{\rmdefault}{\mddefault}{\updefault}{\color[rgb]{0,0,0}8}%
}}}}
\put(6505,-3302){\makebox(0,0)[lb]{\smash{{\SetFigFont{14}{16.8}{\rmdefault}{\mddefault}{\updefault}{\color[rgb]{0,0,0}$i$}%
}}}}
\put(5365,-3287){\makebox(0,0)[lb]{\smash{{\SetFigFont{14}{16.8}{\rmdefault}{\mddefault}{\updefault}{\color[rgb]{0,0,0}$i_1$}%
}}}}
\put(1585,-3227){\makebox(0,0)[lb]{\smash{{\SetFigFont{14}{16.8}{\rmdefault}{\mddefault}{\updefault}{\color[rgb]{0,0,0}$i$}%
}}}}
\put(2407,-5971){\makebox(0,0)[b]{\smash{{\SetFigFont{14}{16.8}{\rmdefault}{\mddefault}{\updefault}{\color[rgb]{0,0,0}8}%
}}}}
\put(592,-3631){\makebox(0,0)[b]{\smash{{\SetFigFont{14}{16.8}{\rmdefault}{\mddefault}{\updefault}{\color[rgb]{0,0,0}5}%
}}}}
\put(592,-4366){\makebox(0,0)[b]{\smash{{\SetFigFont{14}{16.8}{\rmdefault}{\mddefault}{\updefault}{\color[rgb]{0,0,0}5}%
}}}}
\put(1732,-4681){\makebox(0,0)[b]{\smash{{\SetFigFont{14}{16.8}{\rmdefault}{\mddefault}{\updefault}{\color[rgb]{0,0,0}7}%
}}}}
\put(3127,-4696){\makebox(0,0)[b]{\smash{{\SetFigFont{14}{16.8}{\rmdefault}{\mddefault}{\updefault}{\color[rgb]{0,0,0}0}%
}}}}
\put(2722,-2896){\makebox(0,0)[b]{\smash{{\SetFigFont{14}{16.8}{\rmdefault}{\mddefault}{\updefault}{\color[rgb]{0,0,0}9}%
}}}}
\put(3367,-3676){\makebox(0,0)[b]{\smash{{\SetFigFont{14}{16.8}{\rmdefault}{\mddefault}{\updefault}{\color[rgb]{0,0,0}8}%
}}}}
\put(1747,-3646){\makebox(0,0)[b]{\smash{{\SetFigFont{14}{16.8}{\rmdefault}{\mddefault}{\updefault}{\color[rgb]{0,0,0}5}%
}}}}
\put(1342,-3466){\makebox(0,0)[b]{\smash{{\SetFigFont{14}{16.8}{\rmdefault}{\mddefault}{\updefault}{\color[rgb]{0,0,0}3}%
}}}}
\put(1207,-4006){\makebox(0,0)[b]{\smash{{\SetFigFont{14}{16.8}{\rmdefault}{\mddefault}{\updefault}{\color[rgb]{0,0,0}3}%
}}}}
\put(2017,-3481){\makebox(0,0)[b]{\smash{{\SetFigFont{14}{16.8}{\rmdefault}{\mddefault}{\updefault}{\color[rgb]{0,0,0}3}%
}}}}
\put(2092,-4171){\makebox(0,0)[b]{\smash{{\SetFigFont{14}{16.8}{\rmdefault}{\mddefault}{\updefault}{\color[rgb]{0,0,0}2}%
}}}}
\put(2935,-4367){\makebox(0,0)[lb]{\smash{{\SetFigFont{14}{16.8}{\rmdefault}{\mddefault}{\updefault}{\color[rgb]{0,0,0}$i'$}%
}}}}
\put(7870,-4367){\makebox(0,0)[lb]{\smash{{\SetFigFont{14}{16.8}{\rmdefault}{\mddefault}{\updefault}{\color[rgb]{0,0,0}$i'$}%
}}}}
\put(1726,-6406){\makebox(0,0)[lb]{\smash{{\SetFigFont{14}{16.8}{\rmdefault}{\mddefault}{\updefault}{\color[rgb]{0,0,0}(a)}%
}}}}
\put(6901,-6391){\makebox(0,0)[lb]{\smash{{\SetFigFont{14}{16.8}{\rmdefault}{\mddefault}{\updefault}{\color[rgb]{0,0,0}(b)}%
}}}}
\put(325,-3242){\makebox(0,0)[lb]{\smash{{\SetFigFont{14}{16.8}{\rmdefault}{\mddefault}{\updefault}{\color[rgb]{0,0,0}$i_1$}%
}}}}
\put(325,-4727){\makebox(0,0)[lb]{\smash{{\SetFigFont{14}{16.8}{\rmdefault}{\mddefault}{\updefault}{\color[rgb]{0,0,0}$i_2$}%
}}}}
\put(1450,-5042){\makebox(0,0)[lb]{\smash{{\SetFigFont{14}{16.8}{\rmdefault}{\mddefault}{\updefault}{\color[rgb]{0,0,0}$i_3$}%
}}}}
\put(5350,-4592){\makebox(0,0)[lb]{\smash{{\SetFigFont{14}{16.8}{\rmdefault}{\mddefault}{\updefault}{\color[rgb]{0,0,0}$i_2$}%
}}}}
\put(6355,-4922){\makebox(0,0)[lb]{\smash{{\SetFigFont{14}{16.8}{\rmdefault}{\mddefault}{\updefault}{\color[rgb]{0,0,0}$i_3$}%
}}}}
\end{picture}%

%% file: locgapfp.pstex_t
\begin{picture}(0,0)%
\includegraphics{locgapfp.pstex}%
\end{picture}%
\setlength{\unitlength}{3947sp}%
\begingroup\makeatletter\ifx\SetFigFont\undefined%
\gdef\SetFigFont#1#2#3#4#5{%
  \reset@font\fontsize{#1}{#2pt}%
  \fontfamily{#3}\fontseries{#4}\fontshape{#5}%
  \selectfont}%
\fi\endgroup%
\begin{picture}(4977,2895)(1411,-3859)
\put(4072,-2445){\makebox(0,0)[lb]{\smash{{\SetFigFont{11}{13.2}{\rmdefault}{\mddefault}{\updefault}{\color[rgb]{0,0,0}.}%
}}}}
\put(5356,-3781){\makebox(0,0)[lb]{\smash{{\SetFigFont{12}{14.4}{\rmdefault}{\mddefault}{\updefault}{\color[rgb]{0,0,0}$\mathcal{D}_M$}%
}}}}
\put(2866,-1951){\makebox(0,0)[lb]{\smash{{\SetFigFont{12}{14.4}{\rmdefault}{\mddefault}{\updefault}{\color[rgb]{0,0,0}$\mathcal{D}_2$}%
}}}}
\put(2626,-3511){\makebox(0,0)[lb]{\smash{{\SetFigFont{12}{14.4}{\rmdefault}{\mddefault}{\updefault}{\color[rgb]{0,0,0}1}%
}}}}
\put(3751,-3061){\makebox(0,0)[lb]{\smash{{\SetFigFont{12}{14.4}{\rmdefault}{\mddefault}{\updefault}{\color[rgb]{0,0,0}1}%
}}}}
\put(4276,-3136){\makebox(0,0)[lb]{\smash{{\SetFigFont{12}{14.4}{\rmdefault}{\mddefault}{\updefault}{\color[rgb]{0,0,0}1}%
}}}}
\put(4501,-3586){\makebox(0,0)[lb]{\smash{{\SetFigFont{12}{14.4}{\rmdefault}{\mddefault}{\updefault}{\color[rgb]{0,0,0}1}%
}}}}
\put(3451,-2986){\makebox(0,0)[lb]{\smash{{\SetFigFont{12}{14.4}{\rmdefault}{\mddefault}{\updefault}{\color[rgb]{0,0,0}1}%
}}}}
\put(2926,-3136){\makebox(0,0)[lb]{\smash{{\SetFigFont{12}{14.4}{\rmdefault}{\mddefault}{\updefault}{\color[rgb]{0,0,0}1}%
}}}}
\put(4051,-1786){\makebox(0,0)[lb]{\smash{{\SetFigFont{12}{14.4}{\rmdefault}{\mddefault}{\updefault}{\color[rgb]{0,0,0}$M$}%
}}}}
\put(5326,-2761){\makebox(0,0)[lb]{\smash{{\SetFigFont{12}{14.4}{\rmdefault}{\mddefault}{\updefault}{\color[rgb]{0,0,0}$M$}%
}}}}
\put(1516,-1351){\makebox(0,0)[lb]{\smash{{\SetFigFont{12}{14.4}{\rmdefault}{\mddefault}{\updefault}{\color[rgb]{0,0,0}$s_1$}%
}}}}
\put(3901,-1216){\makebox(0,0)[lb]{\smash{{\SetFigFont{12}{14.4}{\rmdefault}{\mddefault}{\updefault}{\color[rgb]{0,0,0}$s_2$}%
}}}}
\put(6076,-2911){\makebox(0,0)[lb]{\smash{{\SetFigFont{12}{14.4}{\rmdefault}{\mddefault}{\updefault}{\color[rgb]{0,0,0}$s_M$}%
}}}}
\put(3391,-3616){\makebox(0,0)[lb]{\smash{{\SetFigFont{12}{14.4}{\rmdefault}{\mddefault}{\updefault}{\color[rgb]{0,0,0}$o$}%
}}}}
\put(1456,-2746){\makebox(0,0)[lb]{\smash{{\SetFigFont{12}{14.4}{\rmdefault}{\mddefault}{\updefault}{\color[rgb]{0,0,0}$\mathcal{D}_1$}%
}}}}
\put(1426,-2161){\makebox(0,0)[lb]{\smash{{\SetFigFont{12}{14.4}{\rmdefault}{\mddefault}{\updefault}{\color[rgb]{0,0,0}$M$}%
}}}}
\end{picture}%

%% file: locgapf0.pstex_t
\begin{picture}(0,0)%
\includegraphics{locgapf0.pstex}%
\end{picture}%
\setlength{\unitlength}{3947sp}%
\begingroup\makeatletter\ifx\SetFigFont\undefined%
\gdef\SetFigFont#1#2#3#4#5{%
  \reset@font\fontsize{#1}{#2pt}%
  \fontfamily{#3}\fontseries{#4}\fontshape{#5}%
  \selectfont}%
\fi\endgroup%
\begin{picture}(4365,4175)(1048,-4528)
\put(5128,-1524){\makebox(0,0)[lb]{\smash{{\SetFigFont{12}{14.4}{\rmdefault}{\mddefault}{\updefault}{\color[rgb]{0,0,0}$s_0$}%
}}}}
\put(1588,-1524){\makebox(0,0)[lb]{\smash{{\SetFigFont{12}{14.4}{\rmdefault}{\mddefault}{\updefault}{\color[rgb]{0,0,0}$s_{k-1}$}%
}}}}
\put(4288,-564){\makebox(0,0)[lb]{\smash{{\SetFigFont{12}{14.4}{\rmdefault}{\mddefault}{\updefault}{\color[rgb]{0,0,0}$j_0$}%
}}}}
\put(2983,-822){\makebox(0,0)[lb]{\smash{{\SetFigFont{12}{14.4}{\rmdefault}{\mddefault}{\updefault}{\color[rgb]{0,0,0}1}%
}}}}
\put(3883,-807){\makebox(0,0)[lb]{\smash{{\SetFigFont{12}{14.4}{\rmdefault}{\mddefault}{\updefault}{\color[rgb]{0,0,0}1}%
}}}}
\put(1783,-957){\makebox(0,0)[lb]{\smash{{\SetFigFont{12}{14.4}{\rmdefault}{\mddefault}{\updefault}{\color[rgb]{0,0,0}$k-\e$}%
}}}}
\put(4753,-972){\makebox(0,0)[lb]{\smash{{\SetFigFont{12}{14.4}{\rmdefault}{\mddefault}{\updefault}{\color[rgb]{0,0,0}$k-\e$}%
}}}}
\put(5218,-2157){\makebox(0,0)[lb]{\smash{{\SetFigFont{12}{14.4}{\rmdefault}{\mddefault}{\updefault}{\color[rgb]{0,0,0}$k-\e$}%
}}}}
\put(5308,-2529){\makebox(0,0)[lb]{\smash{{\SetFigFont{12}{14.4}{\rmdefault}{\mddefault}{\updefault}{\color[rgb]{0,0,0}$j_1$}%
}}}}
\put(3388,-594){\makebox(0,0)[lb]{\smash{{\SetFigFont{12}{14.4}{\rmdefault}{\mddefault}{\updefault}{\color[rgb]{0,0,0}$o_0$}%
}}}}
\put(1603,-3459){\makebox(0,0)[lb]{\smash{{\SetFigFont{12}{14.4}{\rmdefault}{\mddefault}{\updefault}{\color[rgb]{0,0,0}$o_{k-1}$}%
}}}}
\put(5113,-3459){\makebox(0,0)[lb]{\smash{{\SetFigFont{12}{14.4}{\rmdefault}{\mddefault}{\updefault}{\color[rgb]{0,0,0}$o_1$}%
}}}}
\put(1738,-2997){\makebox(0,0)[lb]{\smash{{\SetFigFont{12}{14.4}{\rmdefault}{\mddefault}{\updefault}{\color[rgb]{0,0,0}1}%
}}}}
\put(5053,-2997){\makebox(0,0)[lb]{\smash{{\SetFigFont{12}{14.4}{\rmdefault}{\mddefault}{\updefault}{\color[rgb]{0,0,0}1}%
}}}}
\put(1918,-3882){\makebox(0,0)[lb]{\smash{{\SetFigFont{12}{14.4}{\rmdefault}{\mddefault}{\updefault}{\color[rgb]{0,0,0}1}%
}}}}
\put(4798,-3852){\makebox(0,0)[lb]{\smash{{\SetFigFont{12}{14.4}{\rmdefault}{\mddefault}{\updefault}{\color[rgb]{0,0,0}1}%
}}}}
\put(2173,-564){\makebox(0,0)[lb]{\smash{{\SetFigFont{12}{14.4}{\rmdefault}{\mddefault}{\updefault}{\color[rgb]{0,0,0}$j_{2k-1}$}%
}}}}
\put(1123,-2157){\makebox(0,0)[lb]{\smash{{\SetFigFont{12}{14.4}{\rmdefault}{\mddefault}{\updefault}{\color[rgb]{0,0,0}$k-\e$}%
}}}}
\put(1063,-2514){\makebox(0,0)[lb]{\smash{{\SetFigFont{12}{14.4}{\rmdefault}{\mddefault}{\updefault}{\color[rgb]{0,0,0}$j_{2k-2}$}%
}}}}
\put(1513,-4329){\makebox(0,0)[lb]{\smash{{\SetFigFont{12}{14.4}{\rmdefault}{\mddefault}{\updefault}{\color[rgb]{0,0,0}$j_{2k-3}$}%
}}}}
\put(4798,-4314){\makebox(0,0)[lb]{\smash{{\SetFigFont{12}{14.4}{\rmdefault}{\mddefault}{\updefault}{\color[rgb]{0,0,0}$j_2$}%
}}}}
\end{picture}%

%% file: newlbfl-arxiv.bbl
\begin{thebibliography}{10}
\bibitem{AryaGKMMP01}
V.~Arya, N.~Garg, R.~Khandekar, A.~Meyerson, K.~Munagala, and V.~Pandit.
\newblock Local search heuristics for $k$-median and facility location
  problems.
\newblock {\em SIAM Journal on Computing}, 33(3):544--562, 2004.

\bibitem{CharikarG99}
M.~Charikar and S.~Guha.
\newblock Improved combinatorial algorithms for facility location problems.
\newblock {\em SIAM Journal on Computing}, 34(4):803--824, 2005.

\bibitem{GuhaMM01b}
S.~Guha, A.~Meyerson, and K.~Munagala.
\newblock A constant factor approximation for the single sink edge installation
  problem.
\newblock {\em SIAM Journal on Computing}, 38(6):2246--2442, 2009.

\bibitem{GuhaMM00b}
S.~Guha, A.~Meyerson, and K.~Munagala.
\newblock Facility location with demand dependent costs and generalized clustering.
\newblock {\em Manuscript}, 2000.

\bibitem{GuhaMM00}
S.~Guha, A.~Meyerson, and K.~Munagala.
\newblock Hierarchical placement and network design problems.
\newblock In {\em Proceedings of the 41st Annual {IEEE} Symposium on Foundations of
  Computer Science}, pages 603--612, 2000.

\bibitem{HajiaghayiMM03}
M.~Hajiaghayi, M.~Mahdian, and V.~Mirrokni.
\newblock The facility location problem with general cost functions. 
\newblock {\em Networks}, 42:42--47, 2003.

\bibitem{HardyLP52}
G.~H. Hardy, J.~E. Littlewood, and G.~P{\'o}lya.
\newblock {\em Inequalities}.
\newblock Cambridge University Press, 1952.

\bibitem{KargerM00}
D.~R. Karger and M.~Minkoff.
\newblock Building {S}teiner trees with incomplete global knowledge.
\newblock In {\em Proceedings of the 41st Annual {IEEE} Symposium on
  Foundations of Computer Science}, pages 613--623, 2000.

\bibitem{KorupoluPR00}
M.~R. Korupolu, C.~G. Plaxton, and R.~Rajaraman.
\newblock Analysis of a local search heuristic for facility location problems.
\newblock {\em Journal of Algorithms}, 37(1):146--188, 2000.

\bibitem{Li11}
S.~Li. 
\newblock A 1.488 approximation algorithm for the uncapacitated facility location
problem. 
\newblock In {\em Proceedings of the 38th International Colloquium on Automata Languages
and Programming}, pages 77--88, 2011.

\bibitem{LimWX06}
A.~Lim, F.~Wang, and Z.~Xu. 
\newblock A transportation problem with minimum quantity commitment. 
\newblock {\em Transportation Science}, 40(1):117--129, 2006.

\bibitem{MahdianP03}
M.~Mahdian and M.~P\'{a}l. 
\newblock Universal facility location. 
\newblock In {\em Proceedings of 11th ESA}, pages 409--421, 2003.

\bibitem{MirchandaniF90}
P.~Mirchandani and R.~{Francis, editors}.
\newblock {\em Discrete {L}ocation {T}heory}.
\newblock John Wiley and Sons, Inc., New York, 1990.

\bibitem{Shmoys04}
D.~B. Shmoys.
\newblock The design and analysis of approximation algorithms: facility location as a case
study. 
\newblock In S.~Hosten, J.~Lee, and R.~Thomas, editors. 
\newblock {\em Trends in Optimization, AMS Proceedings of Symposia in Applied Mathematics
  61}, pages 85--97, 2004.

\bibitem{ShmoysTA97}
D.~B. Shmoys, {\'E}.~Tardos, and K.~I. Aardal.
\newblock Approximation algorithms for facility location problems.
\newblock In {\em Proceedings of the 29th Annual {ACM} Symposium on Theory of
  Computing}, pages 265--274, 1997.

\bibitem{Skutella06}
M.~Skutella.
\newblock List scheduling in order of $\al$-points on a single machine.
\newblock In E.~Bampis, K.~Jansen, and C.~Kenyon, editors.
\newblock {\em Efficient Approximation and Online Algorithms: Recent Progress on Classical
Combinatorial Optimization Problems and New Applications}, pages 250--291,
Springer-Verlag, Berlin, 2006. 

\bibitem{Sviridenko02}
M.~Sviridenko.
\newblock An improved approximation algorithm for the metric uncapacitated
  facility location problem.
\newblock In {\em Proceedings of 9th IPCO}, pages 240--257, 2002.

\bibitem{Svitkina08}
Z.~Svitkina.
\newblock Lower-bounded facility location.
\newblock {\em Transactions on Algorithms}, 6(4), 2010.

\bibitem{ZhangCY03}
J.~Zhang, B.~Chen, and Y.~Ye.
\newblock A multi-exchange local search algorithm for the capacitated
  facility location problem.
\newblock {\em Mathematics of Operations Research}, 30:389--403, 2005.

\end{thebibliography}
